\documentclass[10pt,twocolumn,letterpaper]{article}
\usepackage[usenames, dvipsnames]{xcolor}
\usepackage{usenix, times,color,graphicx,amsmath,array,amssymb, subcaption}
\usepackage{color}
\usepackage{xcolor}
\usepackage{multirow}
\usepackage{amsmath}
\usepackage{wrapfig}
\usepackage{amsfonts}
\usepackage{mdframed}
\usepackage{amsthm}
\usepackage{listings}
\usepackage{multirow}
\usepackage{xspace}
\usepackage{comment}
\usepackage{rotating}
\usepackage{colortbl}
\usepackage{siunitx}
\usepackage{makecell}
\usepackage{transparent}
\usepackage{tightenum}
\usepackage{adjustbox}
\usepackage{listings}
\usepackage{microtype}
\usepackage{tabulary}
\usepackage{floatflt}
\usepackage{textcomp}
\usepackage[small,compact]{titlesec}
\usepackage{enumitem}
\usepackage{booktabs}
\usepackage[
  colorlinks=true,
  citecolor=red
]{hyperref}
\usepackage[normalem]{ulem}
\usepackage{standalone}
\usepackage{tikz}
\usetikzlibrary{shapes.geometric}
\usetikzlibrary{arrows.meta,arrows}
\usepackage{pgfplots}
\usepackage{lipsum}
\usepackage[available,functional,reproduced]{usenixbadges}

\newcommand\blfootnote[1]{%
  \begingroup
  \renewcommand\thefootnote{}\footnote{#1}%
  \addtocounter{footnote}{-1}%
  \endgroup
}

\pgfplotsset{compat=1.16}

\hypersetup{
  colorlinks,
  linkcolor={red!50!black},
  citecolor={blue!50!black},
  urlcolor={blue!80!black}
}

\newcommand{\squishlist}{
   \begin{list}{$\bullet$}
    { \setlength{\itemsep}{0pt}      \setlength{\parsep}{3pt}
      \setlength{\topsep}{3pt}       \setlength{\partopsep}{0pt}
      \setlength{\leftmargin}{1.0em} \setlength{\labelwidth}{1em}
      \setlength{\labelsep}{0.5em} } }
\newcommand{\squishend}{
    \end{list}  }
    
\newcommand{\tool}[0]{Dorylus\xspace}

\newcommand{\naively}[0]{na\"{i}vely\xspace}

\newcommand{\codeIn}[1]{{\small\texttt{#1}}}

\newcommand{\mysection}[1]{\vspace{-.2em}\section{#1}\vspace{-.2em}}
\newcommand{\mysubsection}[1]{\subsection{#1}\vspace{-.1em}}
\newcommand{\mysubsubsection}[1]{\subsubsection{{\normalsize #1}}\vspace{-.5em}}
\newcommand{\MyPara}[1]{\vspace{.1em}\noindent\textbf{\textit{#1}}~~}

\newtheorem{theorem}{Theorem}
\newtheorem{definition}{Definition}
\newtheorem{lemma}{Lemma}
\newtheorem{proposition}{Proposition}

\newtheorem{innercustomthm}{Theorem}
\newenvironment{customthm}[1]
  {\renewcommand\theinnercustomthm{#1}\innercustomthm}
{\endinnercustomthm}

\newcommand{\squishlistree}{
   \begin{list}{$\bullet$}
    { \setlength{\itemsep}{0pt}      \setlength{\parsep}{0pt}
      \setlength{\topsep}{3pt}       \setlength{\partopsep}{0pt}
      \setlength{\leftmargin}{1em} \setlength{\labelwidth}{1em}
      \setlength{\labelsep}{0.5em} } }

\newcommand{\squishlisttwo}{
   \begin{list}{$\bullet$}
    { \setlength{\itemsep}{0pt}    \setlength{\parsep}{0pt}
      \setlength{\topsep}{0pt}     \setlength{\partopsep}{0pt}
      \setlength{\leftmargin}{2em} \setlength{\labelwidth}{1.5em}
      \setlength{\labelsep}{0.5em} } }

\newcommand{\hx}[1] {{\textcolor{red}{HX: {#1}}}}
\newcommand{\jt}[1] {{#1}}

\newcommand{\eg}{\hbox{\emph{e.g.}}\xspace}
\newcommand{\ie}{\hbox{\emph{i.e.}}\xspace}

\newcommand{\wrt}{\hbox{\emph{w.r.t.}}\xspace}

\begin{document}





\title{\tool: Affordable, Scalable, and Accurate GNN Training with \\Distributed CPU Servers and Serverless Threads\vspace{-1em}}

\author{\rm{John Thorpe}$^{\dag\clubsuit}$\hspace{1.2em}Yifan Qiao$^{\dag\clubsuit}$\hspace{1.2em}Jonathan Eyolfson$^{\dag}$\hspace{1.2em}Shen Teng$^{\dag}$\hspace{1.2em}Guanzhou Hu$^{\dag\ddag}$\hspace{1.2em}Zhihao Jia$^{\S}$\\[.3em]\rm{Jinliang Wei}$^{\ast}$\hspace{1.2em}Keval Vora$^{\flat}$\hspace{1.2em}Ravi Netravali$^{\sharp}$\hspace{1.2em}Miryung Kim$^{\dag}$\hspace{1.2em}Guoqing Harry Xu$^{\dag}$
\\[.3em]
UCLA$^{\dag}$\hspace{1.2em}University of Wisconsin$^{\ddag}$\hspace{1.2em}CMU$^\S$\hspace{1.2em}Google Brain$^\ast$\hspace{1.2em}Simon Fraser$^{\flat}$\hspace{1.2em}Princeton University$^{\sharp}$
}

\interfootnotelinepenalty 100000
\widowpenalty 100000
\clubpenalty 100000
\newfont{\tf}{phvro at 9.5pt}
\newfont{\tft}{phvro at 7.25pt}

\begin{sloppypar}
\maketitle
\setcounter{page}{1}
\section*{Abstract}
\blfootnote{$^\clubsuit$ Contributed equally.}
A graph neural network (GNN) enables deep learning on structured graph data.
There are two major GNN training obstacles: 1) it relies on high-end servers with many GPUs which are expensive to purchase and maintain, and 2) limited memory on GPUs cannot scale to today's billion-edge graphs.
This paper presents \tool{}: a distributed system for training GNNs.
Uniquely, \tool{} can take advantage of \emph{serverless computing} to increase scalability at a low cost.

The key insight guiding our design is \emph{computation separation}.
Computation separation makes it possible to construct a \emph{deep, bounded-asynchronous} pipeline where graph and tensor parallel tasks can fully overlap, effectively hiding the network latency incurred by Lambdas.
With the help of thousands of Lambda threads, \tool scales GNN training to billion-edge graphs.
Currently, for large graphs, CPU servers offer the best performance per dollar over GPU servers.
Just using Lambdas on top of \tool{} offers up to $2.75\times$ more performance-per-dollar than CPU-only servers.
Concretely, \tool{} is $1.22\times$ faster and $4.83\times$ cheaper than GPU servers for massive sparse graphs.
\tool{} is up to $3.8\times$ faster and $10.7\times$ cheaper compared to existing sampling-based systems.

\mysection{Introduction\label{sec:intro}}

\emph{Graph Neural Networks (GNN)}~\cite{kipf-iclr17,neugraph-atc19,scarselli-tnn09,li-iclr16, hyperbolic-gnn-nips19, jia-kdd20} are a family of NNs designed for deep learning on graph structured data~\cite{gnn-review-18,wu-gnn-survey-19}. The most well-known model in this family is the graph convolutional network (GCN)~\cite{kipf-iclr17}, which uses the connectivity structure of the graph as the filter to perform neighborhood mixing. Other models include graph recursive network (GRN)~\cite{ggsnn-iclr16,graph-lstms-tacl17}, graph attention network (GAT)~\cite{residual-gnn-17, gan-iclr18, gaan-uai18}, and graph transformer network (GTN)~\cite{gtn-nips19}. Due to the prevalence of graph datasets, GNNs have gained increasing popularity across diverse domains such as drug discovery~\cite{gnn-recommendation}, chemistry~\cite{duvenaud-nips15}, program analysis~\cite{program-graphs-iclr18, balog-iclr17}, and recommendation systems~\cite{gnn-recommendation-aaai19,gnn-recommender-kdd18}. In fact, GNN is one of the most popular topics in recent AI/ML conferences~\cite{nips-keywords, iclr-keywords}.
GPUs are the de facto platform to train a GNN  due to their ability to provide highly-parallel computations.  While GPUs offer great efficiency for training, they (and their host machines) are expensive to use. 
To train a (small) million-edge graph, recent works such as NeuGraph~\cite{neugraph-atc19} and Roc~\cite{roc-mlsys20} need at least four such machines. A public cloud offers flexible pricing options, but cloud GPU instances still incur a non-trivial cost --- the lowest-configured p3 instance type on AWS has a price of \$3.06/h; training realistic models requires dozens/hundreds of such machines to work 24/7.  While cost is not a concern for big tech firms, it can place a heavy financial burden on small businesses and organizations.


In addition to being expensive GPU's have limited memory, hindering \emph{scalability}.
For context, real-world graphs are routinely \emph{billion-edge} scale~\cite{roy-sosp15} and continue to grow~\cite{gnn-recommender-kdd18}. NeuGraph and Roc enable coordinated use of multiple GPUs to improve scalability (at higher costs), but they remain unable to handle the billion-edge graphs that are commonplace today. 
Two main approaches exist for reducing the costs and improving the scalability of GNN training, but they each introduce new drawbacks:

\squishlist   
\item CPUs face far looser memory restrictions than GPUs, and operate at significantly lower costs. However, CPUs are unable to provide the parallelism in computations that GPUs can, and thus deliver far inferior \emph{efficiency} (or speed).
    
\item Graph sampling techniques select certain vertices and sample their neighbors when gathering data~\cite{graphsage-nips17, gnn-recommender-kdd18}. Sampling techniques improve scalability by considering less graph data, and it is a generic technique that can be used on either GPU or CPU platforms. However, our experiments (\S\ref{sec:comparison}) and prior work~\cite{roc-mlsys20} highlight two limitations with graph sampling: (1) sampling must be done repeatedly per epoch, incurring time overheads and (2) sampling typically reduces \emph{accuracy} of the trained GNNs. Furthermore, although sampling-based training converges often in practice, there is no guarantee for trivial sampling methods~\cite{chen-gcnvr-pmlr18}. 
\squishend

\MyPara{Affordable, Scalable, and Accurate GNN Training.} This paper devises a \emph{low-cost} training framework for GNNs on \emph{billion-edge graphs}.
Our goal is to simultaneously deliver high efficiency (\eg, close to GPUs) and high accuracy (\eg, higher than sampling).
Scaling to billion-edge graphs is crucial for applicability to real-world use cases.
Ensuring low costs and practical performance improves the accessibility for small organizations and domain experts to make the most out of their rich graph data.

To achieve these goals, we turn to the \emph{serverless computing} paradigm, which has gained increasing traction~\cite{excamera-nsdi17,pocket-osdi18,occupy-cloud-socc17} in recent years through platforms such as AWS Lambda, Google Cloud Functions, or Azure Functions. Serverless computing provides large numbers of parallel ``cloud  function'' threads, or \textit{Lambdas}, at an extremely low price (\ie, \$0.20 for launching one million threads on AWS~\cite{lambda-price}). 
Furthermore, Lambda presents a pay-only-for-what-you-use model, which is much more appealing than dedicated servers for applications that need only massive parallelism.  


Although it appears that Lambdas could be used to complement CPU servers without significantly increasing costs, Lambdas were built to execute light asynchronous tasks, presenting two challenges for NN training: 

\squishlist   
\item \emph{Limited compute resources} (\eg, 2 weak vCPUs)
\item \emph{Restricted network resources} (\eg, 200 Mbps between Lambda servers and standard EC2 servers~\cite{serverless-storage-atc18})
\squishend

A neural network makes heavy use of (linear algebra based) tensor kernels. A Lambda\footnote{We use ``Lambda'' in this paper because our implementation is on AWS while our idea is generally applicable to all types of serverless threads.} thread is often too weak to execute a tensor kernel on large data; breaking the data to tiny minibatches may mitigate the compute problem, but gets penalized by the high data-transfer overhead.  Consequently, using Lambdas \naively for training an NN could result in significant slowdowns (\eg, 21$\times$ slowdowns for training of multi-layer perceptron NNs~\cite{serverless-study-cidr19}, even compared to CPUs).

\MyPara{\tool.}  To overcome these weaknesses, we develop \tool,\footnote{\tool is a genus of army ants that form large marching columns.
} a distributed system that uses cheap CPU servers and Lambda threads to achieve the aforementioned goals for GNN training. We next discuss how \tool leverages GNN's special computation model to overcome the two challenges associated with the use of Lambdas.

The \emph{first} challenge is \emph{how to make computation fit into Lambda's weak compute profile?}  We observed: \emph{not all operations in GNN training need Lambda's parallelism}. GNN training comprises of two classes of tasks~\cite{neugraph-atc19} -- neighbor propagations (\eg, \codeIn{Gather} and \codeIn{Scatter}) over the input graph and per-vertex/edge NN operations (such as \codeIn{Apply}) over the tensor data (\eg, features and parameters). Training a GNN over a large graph is dominated by \emph{graph computation} (see \S\ref{sec:breakdown}), not \emph{tensor computation} that exhibits strong SIMD behaviors and benefits the most from massive parallelism. 

Based on this observation, we divide a training pipeline into a set of \emph{fine-grained tasks} (Figure~\ref{fig:backprop}, \S\ref{sec:pipeline}) based on the type of data they process. Tasks that operate over the graph structure belong to a \emph{graph-parallel path}, executed by CPU instances, while those that process tensor data are in a \emph{tensor-parallel path}, executed by Lambdas. 
Since the graph structure is taken out of tensors (\ie, it is no longer represented as a matrix), the amount of tensor data and computation can be significantly reduced, providing an opportunity for each tensor-parallel task to run a \emph{lightweight linear algebra operation} on \emph{a data chunk of a small size} --- a granularity that a Lambda is capable of executing quickly.

Note that Lambdas are a perfect fit to GNNs' tensor computations. While one could also employ regular CPU instances for compute, using such instances would incur a much higher monetary cost to provide the same level of burst parallelism (\eg, \textbf{2.2$\times$} in our experiments) since users not only pay for the compute but also other unneeded resources (\eg, storage).


The \emph{second} challenge is \emph{how to minimize the negative impact of Lambda's network latency?} Our experiments show that Lambdas can spend one-third of their time on communication. To not let communication bottleneck training, \tool employs a novel parallel computation model, referred to as \emph{bounded pipeline asynchronous computation} (BPAC). BPAC makes full use of \emph{pipelining} where different fine-grained tasks overlap with each other, \eg, when graph-parallel tasks process graph data on CPUs, tensor-parallel tasks process tensor data, simultaneously, with Lambdas. Although pipelining has been used in prior work~\cite{roc-mlsys20, pipedream-sosp19}, in the setting of GNN training, pipelining would be impossible without fine-grained tasks, which are, in turn, enabled by computation separation. 


To further reduce the wait time between tasks, BPAC incorporates \emph{asynchrony} into the pipeline so that a fast task does not have to wait until a slow task finishes even if data dependencies exist between them. Although asynchronous processing has been widely used in the past, \tool faces a unique technical difficulty that no other systems have dealt with: as \tool has two computation paths, where exactly should asynchrony be introduced?

\tool uses asynchrony in a novel way at two distinct locations where staleness can be tolerated: \emph{parameter updates} (in the tensor-parallel path) and \emph{data gathering from neighbor vertices} (in the graph-parallel path). To not let asynchrony slow down the convergence, \tool \emph{bounds} the degree of asynchrony at each location using different approaches (\S\ref{sec:bounded}): \emph{weight stashing}~\cite{pipedream-sosp19} at parameter updates and \emph{bounded staleness} at data gathering. We have formally proved the convergence of our asynchronous model in \S\ref{sec:bounded}.

    


\MyPara{Results.} We have implemented two popular GNNs -- GCN and GAT -- on \tool and trained them over four real-world graphs: \codeIn{Friendster} (3.6B edges), \codeIn{Reddit-full} (1.3B), \codeIn{Amazon} (313.9M), and \codeIn{Reddit-small} (114.8M). With the help of 32 graph servers and thousands of Lambda threads, \tool was able to train a GCN, \emph{for the first time without sampling}, over billion-edge graphs such as \codeIn{Friendster}. 

To enable direct comparisons among different platforms, we built new GPU- and CPU-based training backends based on \tool' distributed architecture (with computation separation). Across our graphs, 
\tool's performance is \textbf{2.05$\times$} and \textbf{1.83$\times$} higher than that of GPU-only and CPU-only servers \emph{under the same monetary budget}.
Sampling is surprisingly slow --- to reach the same accuracy target, it is \textbf{2.62$\times$} slower than \tool due to its slow accuracy climbing. 
In terms of accuracy, \tool can train a model with an accuracy \textbf{1.05$\times$} higher than sampling-based techniques.

\paragraph{Key Takeaway.}
Prior work has demonstrated that Lambdas can only achieve suboptimal performance for DNN training due to the limited compute resources on a Lambda and the extra overheads to transfer model parameters/gradients between Lambdas.
Through computation separation, \tool makes it possible, \emph{for the first time}, for Lambdas to provide a scalable, efficient, and low-cost distributed computing scheme for GNN training.

\tool is useful in two scenarios. First, for small organizations that have tight cost constraints, \tool provides an affordable solution by exploiting Lambdas at an extremely low price. Second, for those who need to train GNNs on very large graphs, \tool provides a scalable solution that supports fast and accurate GNN training on billion-edge graphs. 


\mysection{Background\label{sec:background}}
A GNN takes graph-structured data as input, where each vertex is associated with a feature vector, and outputs a feature vector for each individual vertex or the whole graph. The output feature vectors can then be used by various downstream tasks, such as, graph or vertex classification. By combining the feature vectors and the graph structure, GNNs are able to learn the patterns and relationships among the data, rather than relying solely 
on the features of a single data point.

GNN training combines graph propagation (\eg, \codeIn{Gather} and \codeIn{Scatter}) and NN computations. Prior work~\cite{dgl,wang-dgl19} discovered that GNN development can be made much easier with a programming model that provides a \emph{graph-parallel} interface, which allows programmers to develop the NN with familiar graph operations. A typical example is the deep graph library (DGL)~\cite{dgl}, which unifies a variety of GNN models with a common GAS-like interface.


\MyPara{Forward Pass.} To illustrate, consider graph convolutional network (GCN) as an example. GCN is the simplest
and yet most popular model in the GNN family, with  the following forward propagation rule for the $L$-th layer~\cite{kipf-iclr17}: \\[-.8em]
$$(R1)~~~~~~~H_{L+1}~=~ \sigma(\hat{A}H_{L}W_{L})$$ \\[-1.52em]
$A$ is the adjacency matrix of the input graph, and $\Tilde{A}~=~A~+~I_N$ is the adjacency matrix with self-loops constructed by adding $A$ with $I_N$, the identity matrix. $\Tilde{D}$ is a diagonal matrix such that $\Tilde{D}_{ii}~=~\Sigma_j \Tilde{A}_{ij}$. With $\Tilde{D}$, we can construct a \emph{normalized adjacency matrix}, represented by $\hat{A} = \Tilde{D}^{-\frac{1}{2}}\Tilde{A}\Tilde{D}^{-\frac{1}{2}}$. 
$W_{L}$ is a layer-specific trainable weight matrix.  $\sigma(.)$ denotes a non-linear activation function, such as $\mathit{ReLU}$. $H_{L}$ is the \emph{activations matrix} of the $L$-th layer; $\mathit{H}_{0}$ = $X$ is the input feature matrix for all vertices. 

\begin{figure}
\begin{tabular}{ll}
\begin{minipage}[l]{0.45\linewidth}
\includegraphics[scale=0.45]{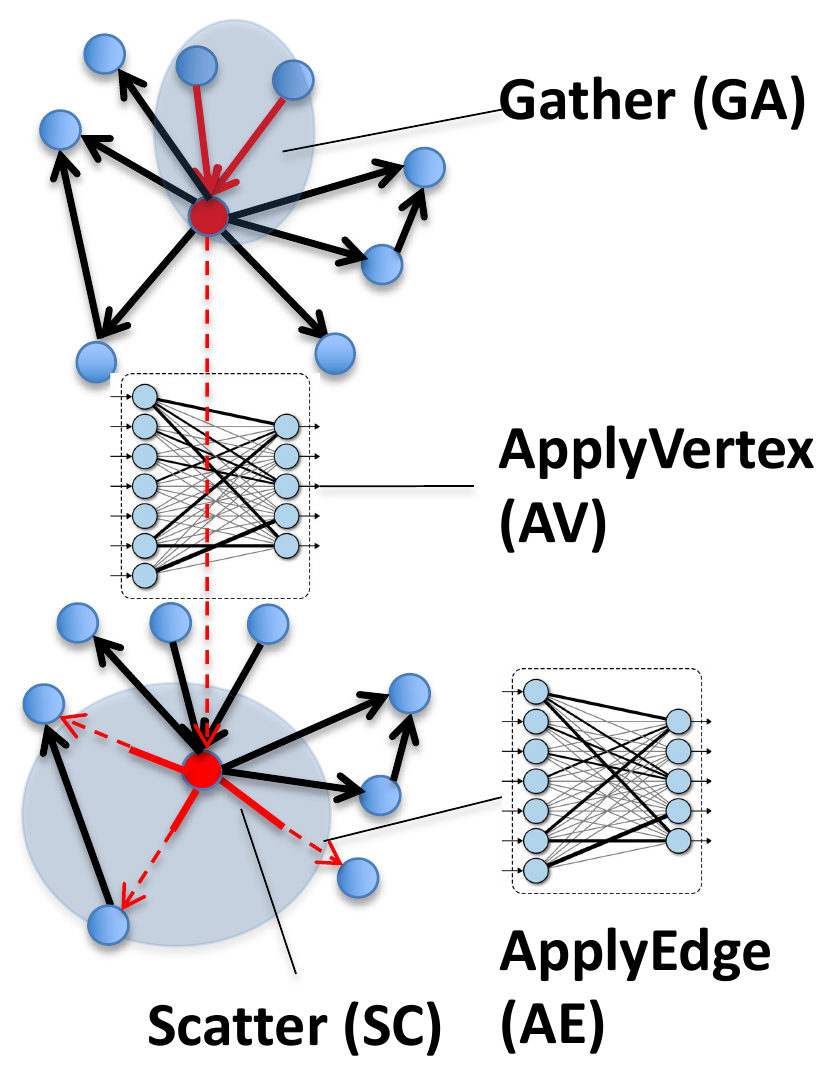}
\end{minipage}
&
\begin{minipage}[l]{0.45\linewidth}
\includegraphics[scale=0.38]{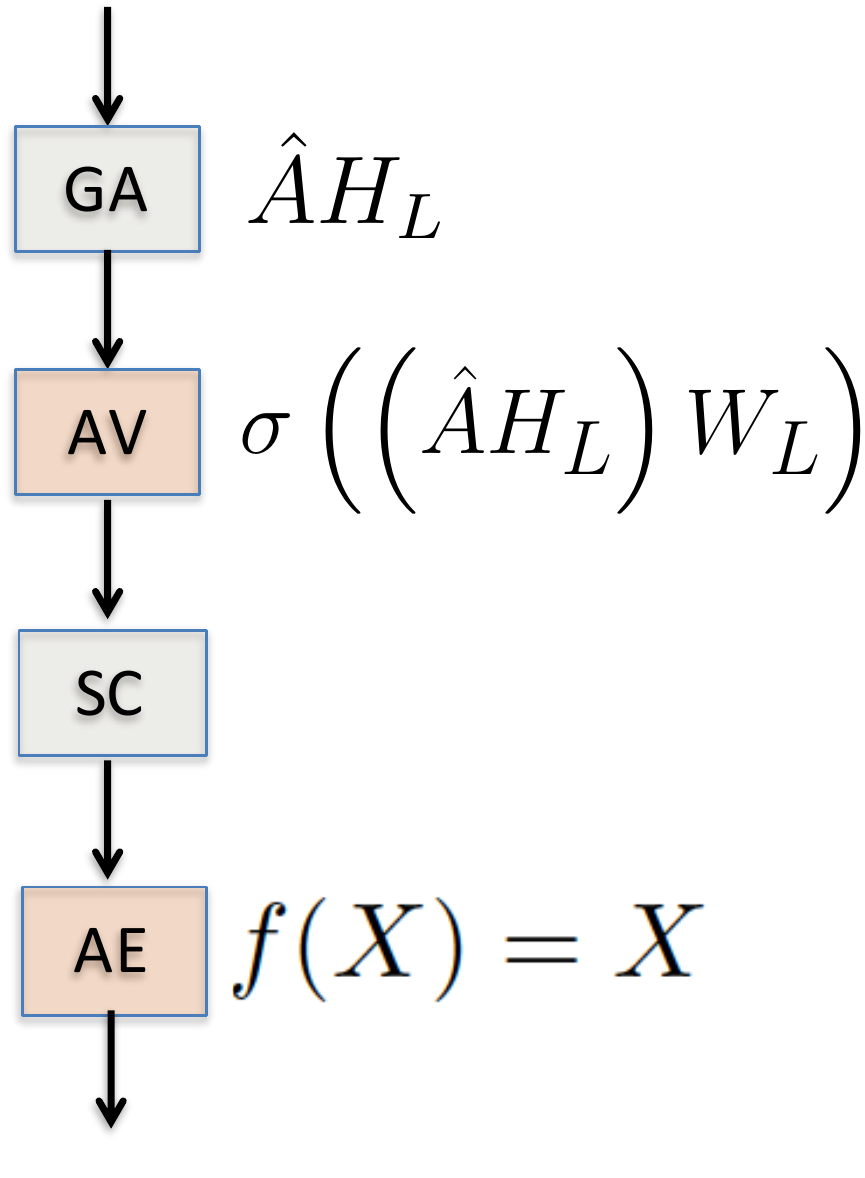}
\end{minipage}
\\
(a) Computation model & (b) Annotated dataflow graph 
\end{tabular}
\caption{A graphical illustration of GCN's computation model and dataflow graph in each forward layer. In (a), edges in {\color{red}red} represent those along which information is being propagated; solid edges represent standard \codeIn{Gather}/\codeIn{Scatter} operations while dashed edges represent NN operations. (b) shows a mapping between the SAGA-NN programming model and the rule R1. \label{fig:computation}}

\vspace{-1.5em}
\end{figure}

Mapping R1 to the vertex-centric computation model is familiar to the systems community~\cite{neugraph-atc19} --- each forward layer has four components: \codeIn{Gather} (GA), \codeIn{ApplyVertex} (AV), \codeIn{Scatter} (SC), and \codeIn{ApplyEdge} (AE), as shown in Figure~\ref{fig:computation}(a). One can think of layer $L$'s activations matrix $\mathit{H}_{L}$ as a group of \emph{activations vectors}, each associated with a vertex (as analogous to \emph{vertex value} in the graph system's terminology). The goal of each forward layer is to compute a new activations vector for each vertex based on the vertex's previous activations vector (which, initially, is its feature vector) and the information received from its in-neighbors. Different from traditional graph processing, the computation of the new activations matrix $\mathit{H}_{L+1}$ is based on computationally intensive NN operations rather than a numerical function.

Figure~\ref{fig:computation}(b) illustrates how these vertex-centric graph operations correspond to various components in R1. First, GA retrieves a vector from each in-edge of a vertex and aggregates these vectors into a new vector $v$. 
In essence, applying GA on all vertices can be implemented as a matrix multiplication 
$\hat{A}H_{L}$, where 
$\hat{A}$ is the normalized adjacency matrix and $H_{L}$ is the input activations matrix.
Second, $(\hat{A}H_{L})$ 
is fed to AV, which performs neural network operations to obtain a new activations matrix $H_{L+1}$.
For GCN, AV multiplies 
$(\hat{A}H_{L})$ with a trainable weight matrix $W_{L}$ and applies a non-linear activation function $\sigma$.
Third, the output of AV goes to SC, which propagates the new activations vector of each vertex along all out-edges of the vertex.
Finally, the new activations vector of each vertex goes into an edge-level NN architecture to compute an activations vector for each edge. For GCN, the edge-level NN is not needed, and hence, AE is an identity function.
We leave AE in the figure for generality as it is needed by other GNN models. 

The output of AE is fed to GA in the next layer. Repeating this process $k$ times (\ie, $k$ layers) allows the vertex to consider features of vertices $k$ hops away. Other GNNs such as GGNNs and GATs have similar computation models, but each varies the method used for aggregation and the NN. 

\MyPara{Backward Pass.} A GNN's backward pass computes the gradients for all trainable weights in the vertex- and edge-level NN architectures (\ie, AV and AE). The backward pass is performed following the chain rule of back propagation. For example, the following rule specifies how to compute the gradients in the first layer for a 2-layer GCN: ~~~~~~~~\\[-.2em]
$$(R2)~~~~~ \nabla_{W_{0}} \mathcal{L} = \left(\hat{A}X\right)^{T}\left[\sigma'\left(\mathit{in}_1\right)\odot\hat{A}^T(Z-Y) W_1^T\right]$$\\[-0.8em]
Here $Z$
is the output of the GCN, $Y$ is the label matrix (\ie, ground truth), $X$ is the input features matrix, $W_i$ is the weight matrix for layer $i$, and $\mathit{in}_1$ = $\hat{A}XW_0$. $\hat{A}^T$ and $W_i^T$ are the transpose of $\hat{A}$ and $W_i$, respectively.

A training \emph{epoch} consists of a forward and a backward pass, followed by {\em weights update}, which uses the gradients computed in the backward pass to update the trainable weights in the vertex- and edge-level NN architectures in a GNN.
The training process runs epochs repeatedly until reaching acceptable accuracy.

\mysection{Design Overview\label{sec:design}}

This section provides an overview of the \tool architecture. The next three sections discuss technical details including how to split training into fine-grained tasks and connect them in a deep pipeline (\S\ref{sec:pipeline}), and how \tool bounds the degree of asynchrony (\S\ref{sec:bounded}), manages and autotunes Lambdas (\S\ref{sec:lambda}). 

Figure~\ref{fig:architecture} depicts \tool's architecture, which is comprised of three major components: EC2 graph servers, Lambda threads for tensor computation, and EC2 parameter servers. An input graph is first partitioned using an edge-cut algorithm~\cite{zhu-osdi16} that takes care of load balancing across partitions. Each partition is hosted by a graph server (GS). 

GSes communicate with each other to execute graph computations by sending/receiving data along cross-partition edges. GSes also communicate with Lambda threads to execute tensor computations. Graph computation is done in a conventional way, breaking a vertex program into vertex-parallel (\eg, \codeIn{Gather}) and edge-parallel stages (\eg, \codeIn{Scatter}). 

Each vertex carries a vector of float values and each edge carries a value of a user-defined type specific to the model. For example, for a GCN, edges do not carry values and \codeIn{ApplyEdge} is an identity function; for a GGNN, each edge has an integer-represented type, with different weights for different edge types. After partitioning, each GS hosts a graph partition where vertex data are represented as a two-dimension array and edge data are represented as a single array. Edges are stored in the \emph{compressed sparse rows} (CSR) format; inverse edges are also maintained for the backpropagation. 

Each GS maintains a \emph{ghost buffer}, storing data that are scattered in from remote servers. Communication between GSes is needed only during \codeIn{Scatter} in both (1) forward pass where activation values are propagated along cross-partition edges and (2) backward pass where gradients are propagated along the same edges in the reverse direction.

Tensor operations such as AV and AE, performed by Lambdas, interleave with graph operations. Once a graph operation finishes, it passes data to a Lambda thread, which employs a high-performance linear algebra kernel for tensor computation. Both the forward and backward passes use Lambdas, which communicate frequently with parameter servers (PS) --- the forward-pass Lambdas retrieve weights from PSes to compute layer outputs, while the backward-pass Lambdas compute updated weights. 


\begin{figure}[ht]
\begin{adjustbox}{max width=\linewidth}
    \centering
    \includegraphics[scale=.5]{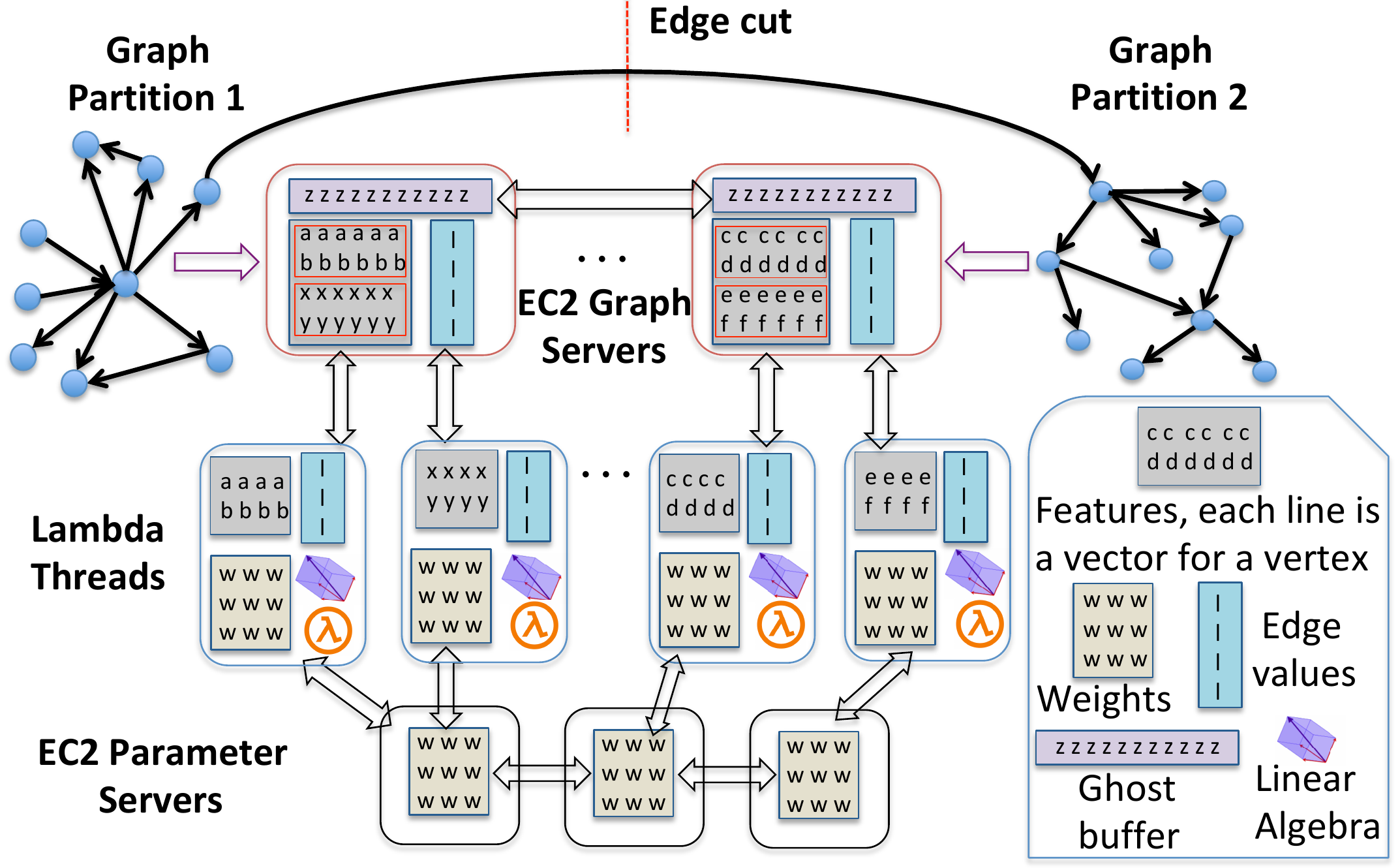}
\end{adjustbox}
    \caption{\tool's architecture. \label{fig:architecture}}
\end{figure}

\mysection{Tasks and Pipelining\label{sec:pipeline}}
\MyPara{Fine-grained Tasks.} As discussed in \S\ref{sec:intro}, the first challenge in using Lambdas for training is to decompose the process into a set of fine-grained tasks that can (1) overlap with each other and (2) be processed by Lambdas' weak compute. In \tool, task decomposition is done based on both \emph{data type} and \emph{computation type}. In general, computations that involve the adjacency matrix of the input graph (\ie, any computation that multiplies any form of the adjacency matrix $A$ with other matrices) are formulated as graph operations performed on GSes, while computations that involve only tensor data can benefit the most from massive parallelism 
and hence run in Lambdas. Next, we discuss specific tasks over each training epoch, which consists of a \emph{forward pass} that computes the output using current weights, followed by a \emph{backward pass} that uses a loss function to compute weight updates.

\textbf{A forward pass} can be naturally divided into four tasks, as shown in Figure~\ref{fig:computation}(a). \codeIn{Gather} (GA) and \codeIn{Scatter} (SC) perform computation over the graph structure; they are thus graph-parallel tasks for execution on GSes.
\codeIn{ApplyVertex} (AV) and \codeIn{ApplyEdge} (AE) multiply matrices involving only features and weights and apply activation functions such as $\mathit{ReLU}$. Hence, they are executed by Lambdas.  

For AV, Lambda threads retrieve vertex data ($\mathit{H}_{L}$ in \S\ref{sec:background}) from GSes and weight data ($\mathit{W}_{L}$) from PSes, compute their product, apply $\mathit{ReLU}$, and send the result back to GSes as the input for \codeIn{Scatter}. When AV returns, SC 
sends data, along cross-partition edges, to the machines that host their destination vertices. 

AE immediately follows SC. To execute AE on an edge, each Lambda thread retrieves (1) vertex data from the source and destination vertices of the edge (\ie, activations vectors), and (2) edge data (such as edge weights) from GSes. It computes a per-edge update by performing model-specific tensor operations. 
These updates are streamed back to GSes and become the inputs of the next layer's GA task.  

\begin{figure}[htb]
    \centering
    \includegraphics[scale=.4]{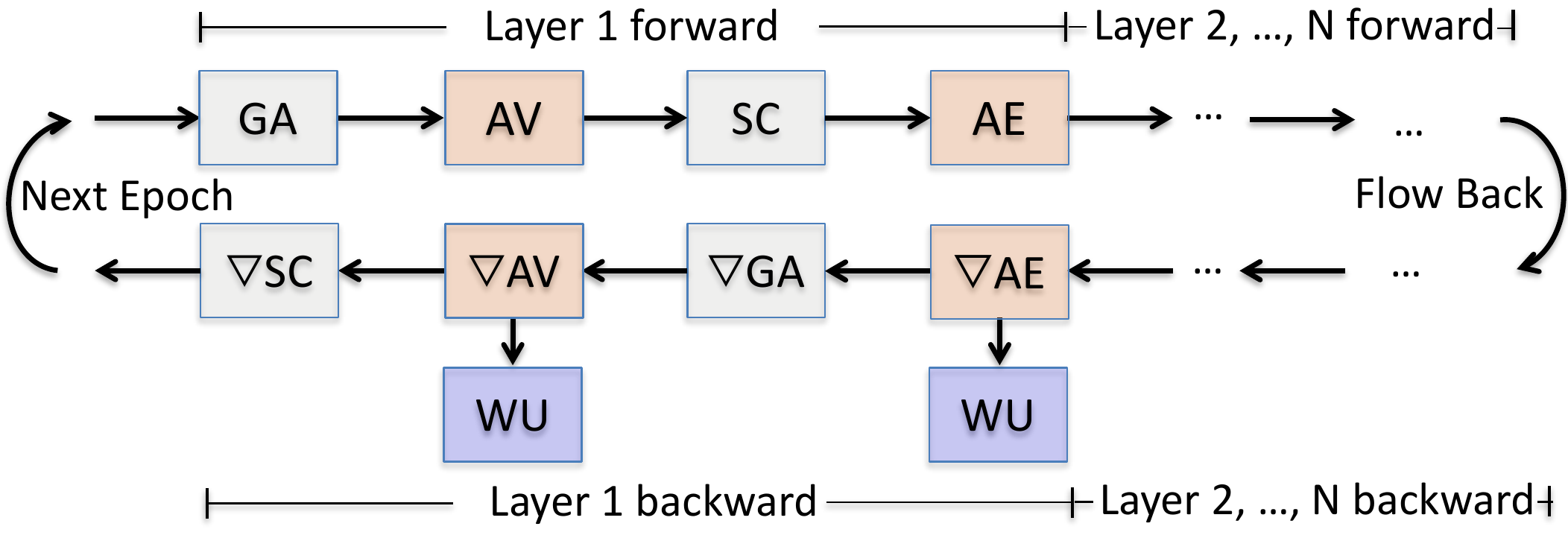}
      \vspace{-2em}
    \caption{\tool's forward and backward dataflow with nine tasks: \codeIn{Gather} (GA) and \codeIn{Scatter} (SC) and their corresponding backward tasks $\triangledown$GA and $\triangledown$SC; \codeIn{ApplyVertex} (AV), \codeIn{ApplyEdge} (AE), and their backward tasks $\triangledown$AV and $\triangledown$AE; the weight update task \codeIn{WeightUpdate} (WU). 
    \label{fig:backprop}}
    \vspace{0em}
\end{figure}

\textbf{A backward pass} involves GSes, Lambdas, and PSes that \emph{coordinate} to run a graph-augmented SGD algorithm, as specified by R2 in \S\ref{sec:background}. For each task in the forward pass, there is a corresponding backward task that either propagates information in the \emph{reverse direction of edges} on the graph or computes the gradients of its trainable weights with respect to a given loss function. Additionally, a backward pass includes \codeIn{WeightUpdate} (WU), which aggregates the gradients across PSes. Figure~\ref{fig:backprop} shows their dataflow.  $\triangledown$GA and $\triangledown$SC are the same as GA and SC except that they propagate information in the reverse direction.  $\triangledown$AE and $\triangledown$AV are the backward tasks for AE and AV, respectively. AE and AV apply weights to compute the output of the edge and vertex NN. Conversely, $\triangledown$AE and $\triangledown$AV compute weight updates for the NNs, which are the inputs to WU.

$\triangledown$AE and $\triangledown$AV perform tensor-only computation and are executed by Lambdas. Similar to the forward pass, GA and SC in the backward pass are executed on GSes. WU performs weights updates and is conducted by PSes.

\MyPara{Pipelining.} In the beginning, vertex and weight data take their initial values (\ie, $H_{0}$ and $W_{0}$), which will change as the training progresses. GSes kick off training by running parallel graph tasks. To establish a full pipeline, \tool divides vertices in each partition into intervals (\ie, minibatches).
For each interval, the amount of tensor computation (done by a Lambda) depends on both the numbers of vertices (\ie, AV) and edges (\ie, AE) in the interval, while the amount of graph computation (on a GS) depends primarily on the number of edges (\ie, GA, and SC).  
To balance work across intervals, our division uses a simple algorithm to ensure that different intervals have the same numbers of vertices and vertices in each interval have similar numbers of inter-interval edges. These edges incur cross-minibatch dependencies that our asynchronous pipeline needs to handle (see \S\ref{sec:bounded}). 

Each interval is processed by a task. When the pipeline is saturated, different tasks will be executed on distinct \emph{intervals} of vertices. Each GS maintains a \emph{task queue} and enqueues a task once it is ready to execute (\ie, its input is available). To fully utilize CPU resources, the GS uses a thread pool where the number of threads equals the number of vCPUs. When the pool has an available thread, the thread retrieves a task from the task queue and executes it. The output of a GS task is fed to a Lambda for tensor computation. 

\begin{figure}[ht]
\begin{adjustbox}{max width=\linewidth}
    \centering
    \includegraphics[scale=.5]{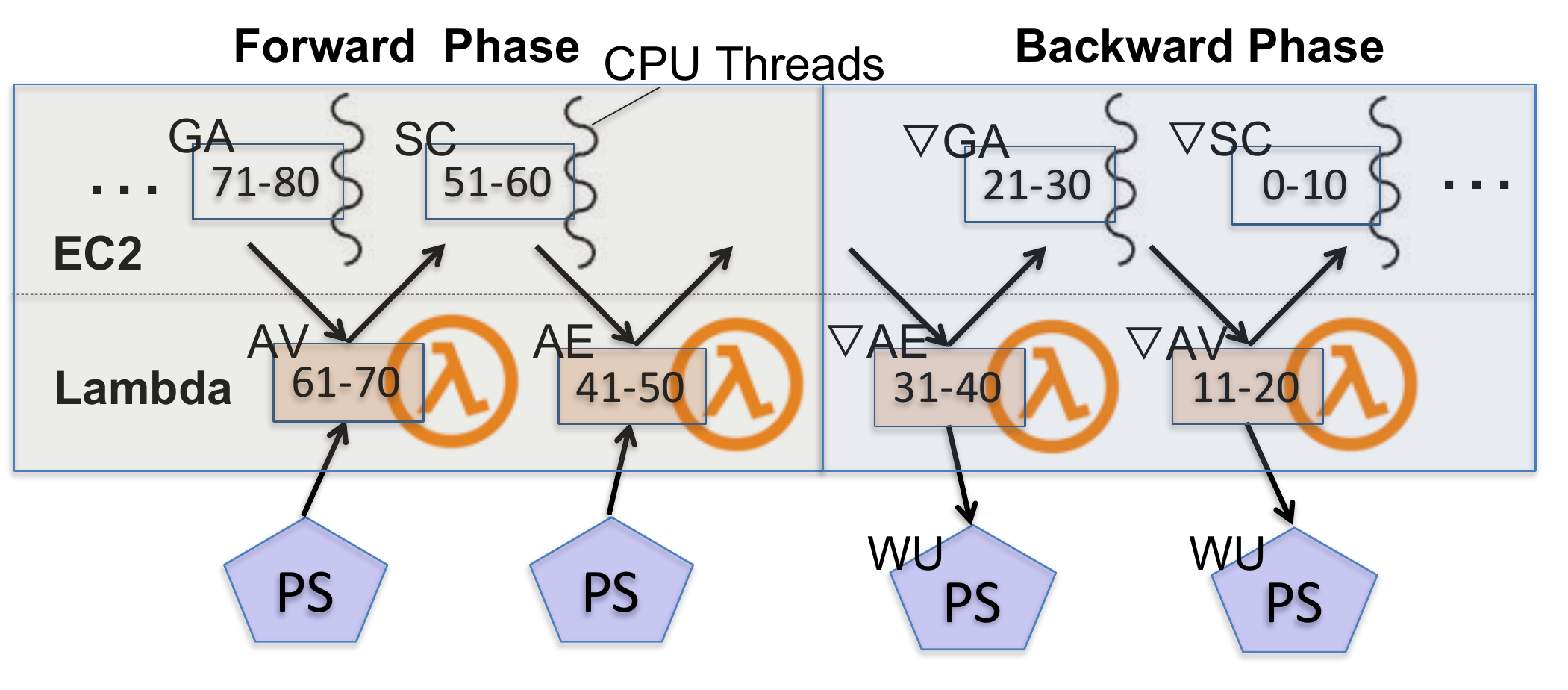}
\end{adjustbox}
    \caption{A \tool pipeline for an epoch: the number range (\eg, 71-80) in each box represents a particular vertex interval (\ie, minibatch); different intervals are at different locations of the pipeline and processed by different processing units: GS, Lambda, or PS.\label{fig:pipeline}}
\end{figure}

Figure~\ref{fig:pipeline} shows a typical training pipeline under \tool. Initially, \tool enqueues a set of GA tasks, each processing a vertex interval. Since the number of threads on each GS is often much smaller than the number of tasks, some tasks finish earlier than others and their results are pushed immediately to Lambda threads for AV. Once they are done, their outputs are sent back to the GS for \codeIn{Scatter}. During a backward phase, both $\triangledown$AE and $\triangledown$AV compute gradients and send them to PSes for weight updates. 



Through effective pipelining, \tool overlaps the graph-parallel and tensor-parallel computations so as to hide Lambdas' communication latency.   
Note that although pipelining is not a new idea, enabling pipelining in GNN training requires fine-grained tasks and the insight of computation separation, which are our unique contributions.

\mysection{Bounded Asynchrony \label{sec:bounded}} 

To unleash the full power of pipelining, \tool performs a unique form of \emph{bounded asynchronous} training so that workers do not need to wait for updates to proceed in most cases. This is paramount for the pipeline's performance especially because Lambdas run in an extremely dynamic environment and stragglers almost always exist. On the other hand, a great deal of evidence~\cite{adam,pipedream-sosp19,poseidon-atc17} shows that asynchrony slows down convergence --- fast-progressing minibatches may use out-of-date weights, prolonging the training time.

Bounded staleness~\cite{cui-atc14, hogwild-nips11} is an effective technique for mitigating the convergence problem by employing lightweight synchronization. However, \tool faces a unique challenge that does not exist in any existing system, that is, there are \emph{two synchronization points} in a \tool pipeline: (1) weight synchronization at each WU task and (2) synchronization of (vertex) activations data from neighbors at each GA. 

\mysubsection{Bounded Asynchrony at Weight Updates}  
To bound the degree of asynchrony for weight updates, we use \emph{weight stashing} proposed in PipeDream~\cite{pipedream-sosp19}. A major reason for slow convergence is that, under full asynchrony, different vertex intervals are at their own training pace; some intervals may use a particular version $v_0$ of weights during a forward pass to compute gradients while applying these gradients on another version $v_1$ of weights on their way back in the backward pass. 
In this case, the weights on which gradients are computed are not those on which they are applied, leading to \emph{statistical inefficiency}. Weight stashing is a simple technique that allows any interval to use the latest version of weights available in a forward pass and stashes (\ie, caches) this version for use during the corresponding backward pass. 

Although weight stashing is not new, applying it in \tool poses unique challenges in the PS design. Weight stashing occurs at PSes, which host weight matrices and perform updates.
To balance loads and bandwidth usage, \tool supports multiple PSes and always directs a Lambda to a PS that has the lightest load. Since Lambdas can be in different stages of an epoch (\eg, the forward and backward passes, and different layers), \tool lets each PS host a replication of weight matrices of \emph{all layers}, making load balancing much easier to do since any Lambda can use any PS in any stage. Note that this design is different from that of traditional PSes~\cite{parameterserver}, each of which hosts parameters for a layer. Since a GNN often has very few layers, replicating all weights would not take much memory and is thus feasible to do at each PS. Clearly, this approach does \emph{not} work for regular DNNs with many layers. 

However, weight stashing significantly complicates this design. A vertex interval can be processed by different Lambdas when it flows to different tasks --- \eg, its AV and AE are executed by different Labmdas, which can retrieve weights from different PSes. Hence, if we allow any Lambda to use any PS, each PS has to maintain not only the latest weight matrices but also their stashed versions for \emph{all} intervals in the graph; this is practically impossible due to its prohibitively high memory requirement. 

To overcome this challenge, we do \emph{not} replicate all weight stashes across PSes. Instead, each PS still contains a replication of all the latest weights but weight stashes only for a \emph{subset} of vertex intervals. For each interval in a given epoch, the interval's weight stashes are only maintained on the first PS it interacts with in the epoch. 
In particular, once a Lambda is launched for an AV task, which is the first task that uses weights in the pipeline, its launching GS picks the PS with the lightest load and notifies the Lambda of its address. Furthermore, the GS remembers this choice for the interval --- when this interval flows to subsequent tensor tasks (\ie, AE, $\triangledown$AV, $\triangledown$AE, and WU), the GS assigns the same PS to their executing Lambdas because the stashed version for this interval only exists on that particular PS in this epoch.  

PSes periodically broadcast their latest weight matrices.

\mysubsection{Bounded Asynchrony at \codeIn{Gather}} 
Asynchronous \codeIn{Gather} allows vertex intervals to progress independently using stale vertex values (\ie, activations vectors) from their neighbors without waiting for their updates. Although asynchrony has been used in a number of graph systems~\cite{aspire,cui-atc14}, these systems perform iterative processing with the assumption that with more iterations they will eventually reach convergence. Different from these systems, the number of layers in a GNN is determined statically and an $n$-layer GNN aims to propagate the impact of a vertex's $n$-hop neighborhood to the vertex. Since the number of layers cannot change during training, an important question that needs be answered is: can asynchrony change the semantics of the GNN? This boils down to two sub-questions: 
(1) Can vertices \emph{eventually} receive the effect of their $n$-hop neighborhood?
(2) Is it possible for any vertex to receive the effect of its $m$-hop neighbor where $m > n$ after many epochs? 
We provide informal correctness/convergence arguments in this subsection and turn to a formal approach in \S\ref{sec:proof}.

\underline{The answer to the first question is \emph{yes}}. This is because the GNN computation is driven by the accuracy of the computed weights, which is, in turn, based on the effects of $n$-hop neighborhoods. To illustrate, consider a simple 2-layer GNN and a vertex $v$ that moves faster in the pipeline than all its neighbors. Assume that by the time $v$ enters the GA of the second layer, none of its neighbors have finished their first-layer SC yet. In this case, the GA task of $v$ uses stale values from its neighbors (\ie, the same as what were used in the previous epoch). This would clearly generate large errors at the end of the epoch. However, in subsequent epochs, the slow-moving neighbors update their values, which are gradually propagated to $v$. 
Hence, the vertex eventually receives the effects of its $n$-hop neighborhood (collectively) across different epochs depending on its neighbors' progress. After each vertex observes the required values from the $n$-hop neighborhood, the target accuracy is reached. 

\underline{The answer to the second question is \emph{no}}. This is because the number of layers determines the \emph{farthest distance} the impact of a vertex can travel despite that training may execute many epochs. When a vertex interval finishes an epoch, it comes back to the initial state where their values are \emph{reset} to their initial feature vectors (\ie, the accumulative effect is cleared). Hence, even though a vertex $v$ progresses asynchronously relative to its neighbors, the neighbors' activation vectors are scattered out in the previous SC and carry the effects of their at most \{$n{-}1$\}-hop neighbors (after which the next GA will cycle back to effects of 1-hop neighbors), which, for vertex $v$, belong strictly in its $n$-hop neighborhood. This means, it is impossible for any vertex to receive the impact of any other vertex that is more than $n$-hops away.

We use \emph{bounded staleness} at \codeIn{Gather} --- a fast-moving vertex interval is allowed to be at most S epochs away from the slowest-moving interval. This means vertices in a given epoch are allowed to use stale vectors from their neighbors only if these vectors are within $S$ epochs away from the current epoch. Bounded staleness allows fast-moving intervals to make quick progress when recent updates are available (for efficiency), but makes them wait when updates are too stale (to avoid launching Lambdas for useless computation). 

\mysubsection{Convergence Guarantee\label{sec:proof}}
Asynchronous weight updates with bounded staleness has been well studied, and its convergence has been proved by \cite{qirong-ssp-nips13}. The convergence of asynchronous \codeIn{Gather} with bounded staleness $S$ is guaranteed by the following theorem:
\begin{theorem}
\label{thm:async-gnn}
Suppose that (1) the activation $\sigma(\cdot)$ is $\rho$-Lipschitz, 
(2) the gradient of the cost function $\nabla_{z}f(y,z)$ is $\rho$-Lipschitz and bounded, 
(3) the gradients for weight updates $\left\| g_{AS}(W) \right\|_\infty$, $\left\| g(W)\right\|_\infty$, and $\left\| \nabla \mathcal{L}(W) \right\|_\infty$ are all bounded by some constant $G > 0$ for all $\hat{A}$, $X$, and $W$, 
(4) the loss $\mathcal{L}(W)$ is $\rho$-smooth\footnote{$\mathcal{L}$ is \textit{$\rho$-Lipschitz smooth} if $\forall W_{1}, W_{2}, |\mathcal{L}(W_2) - \mathcal{L}(W_1) - \left\langle\nabla \mathcal{L}\left(W_{1}\right), W_{2}-W_{1}\right\rangle| \leq \frac{\rho}{2}\left\|W_{2}-W_{1}\right\|_{F}^{2}$, where $\left\langle A, B\right\rangle = tr(A^{T}B)$ is the inner product of matrix $A$ and matrix $B$, and $\left\|A\right\|_{F}$ is the Frobenius norm of matrix $A$.}.  
Then given the local minimizer $W^{*}$, there exists a constant $K > 0$, s.t., $\forall N > L\times S$ where $L$ is the number of layers of the GNN model and $S$ is the staleness bound; if we train a GNN with asynchronous \codeIn{Gather} under a bounded staleness for $R \leq N$ iterations where $R$ is chosen uniformly from $[1, N]$, we will have
\begin{equation*}
\mathbb{E}_{R}
\left\|\nabla \mathcal{L}\left(W_{R}\right)\right\|_{F}^{2} 
\leq 2 \frac{\mathcal{L}\left(W_{1}\right)-\mathcal{L}\left(W^{*}\right)+K+\rho K}{\sqrt{N}},
\end{equation*}
for the updates $W_{i+1} = W_i - \gamma g_{AS}(W_i)$ and the step size $\gamma = min\left\{ \frac{1}{\rho}, \frac{1}{\sqrt{N}}\right\}$.
\end{theorem}
In particular, we have $\lim_{N\rightarrow\infty} \mathbb{E}_{R}\left\|\nabla \mathcal{L}\left(W_{R}\right)\right\|^{2}=0$, indicating that asynchronous GNN training will eventually converge to a local minimum.
The full-blown proof can be found on Appendix \S\ref{sec:theorem}. It mostly follows the convergence proof of the \emph{variance reduction (VR)} algorithm in \cite{chen-gcnvr-pmlr18}. However, our proof differs from \cite{chen-gcnvr-pmlr18} in two major aspects: (1) \tool performs whole-graph training and updates weights only once per layer per epoch, while VR samples the graph and trains on mini-batches and thus it updates weights multiple times per layer per epoch; (2) \tool's asynchronous GNN training can use neighbor activations that are up to $S$-epoch stale, while VR can take only $1$-epoch-stale neighbor activations. Since $S$ is always \emph{bounded} in \tool, the convergence is guaranteed regardless of the value of $S$.

Note that compared to a sampling-based approach, our asynchronous computation is guaranteed to converge. On the contrary, although sampling-based training converges often in practice, there is no guarantee for trivial sampling methods~\cite{chen-gcnvr-pmlr18}, not to mention that sampling incurs a per-epoch overhead and reduces accuracy. 





\mysection{Lambda Management \label{sec:lambda}}  

Each GS runs a Lambda controller, which launches Lambdas, batches data to be sent to each Lambda, monitors each Lambda's health, and routes its result back to the GS. 

Lambda threads are launched by the controller for a task $t$ at the time $t$'s previous task starts executing. For example, \tool launches $n$ Lambda threads, preparing them for AV when $n$ GA tasks start to run. 
Each Lambda runs with OpenBLAS library~\cite{openblas} that is optimized to use AVX instructions for efficient linear algebra operations.
Lambdas communicate with GSes and PSes using ZeroMQ~\cite{zmq}.

All of our Lambdas are deployed inside the virtual private cloud (VPC) rather than public networks to maximize Lambdas' bandwidth when communicating with EC2 instances. When a Lambda is launched, it is given the addresses of its launching GS and a PS. Once initialized, the Lambda initiates communication with the GS and the PS, pulling vertex, edge and weight data from these servers. Since Lambda threads are used throughout the training process, these Lambdas quickly become ``warm'' (\ie, the AWS reuses a container that already has our code deployed instead of cold-starting a new container) and efficient. Our controller also times each Lambda execution and relaunches it after timeout.

\MyPara{Lambda Optimizations.} One significant challenge to overcome is Lambdas' limited network bandwidth~\cite{serverless-study-cidr19, serverless-storage-atc18}. Although AWS has considerably improved Lambdas' network performance~\cite{labmda-network-link}, the per-Lambda bandwidth goes down as the number of Lambdas increases. For example, for each GS, when the number of Lambdas it launches reaches 100, the per-Lambda bandwidth drops to $\thicksim$200Mbps, which is more than 3$\times$ lower than the peak bandwidth we have observed ($\thicksim$800Mbps). We suspect that this is because many Lambdas created by the same user get scheduled on the same machine and share a network link. 

\tool provides three optimizations for Lambdas:

The first optimization is task fusion. Since AV of the last layer in a forward pass is connected directly to $\triangledown$AV of the last layer in the next backward pass (see Figure~\ref{fig:pipeline}), we merge them into a single Lambda-based task, reducing invocations of thousands of Lambdas for each epoch and saving a round-trip communication between Lambdas and GSes.  

The second optimization is tensor rematerialization~\cite{checkmate-mlsys20, dynamic-rematerialize20}. Existing frameworks cache intermediate results during the forward pass as these results can be reused in the backward pass. For GNN training, for instance, $\hat{A}HW$ is such a computation whose result needs to be cached. Here tensor computation is performed by Lambdas while caching has to be done on GSes. Since a Lambda's bandwidth is limited and network communication is a bottleneck, it is more profitable to rematerialize these intermediate tensors by launching more Lambdas rather than retrieving them from GSes. 

The third optimization is Lambda-internal streaming. In particular, if a Lambda is created to process a data chunk, we let the Lambda retrieve the first half of the data, with which it proceeds to computation while simultaneously retrieving the second half. This optimization overlaps computation with communication from within each Lambda, leading to reduced Lambda response time.

\MyPara{Autotuning Numbers of Lambdas.}
Due to inherent dynamism in Lambda executions, it is not feasible to statically determine the number of Lambdas to be used. On the performance side, the effectiveness of Lambdas depends on whether the pipeline can be saturated. In particular, since certain graph tasks (such as SC) rely on results from tensor tasks (such as AV), too few Lambdas would not generate enough task instances for the graph computation $G$ to saturate CPU cores. On the cost side, too many Lambdas \emph{overstaturate} the pipeline --- they can generate too many CPU tasks for the GS to handle. The optimal number of Lambdas is also related to the pace of the graph computation, which, in turn, depends on the graph structure (\eg, density) and partitioning that are hard to predict before execution.

To solve the problem, we develop an autotuner that starts the pipeline by using \codeIn{min}(\#\codeIn{intervals}, 100) as the number of Lambdas where \codeIn{intervals} represents the number of vertex intervals on each GS. Our autotuner auto-adjusts this number by periodically checking the size of the CPU's task queue --- if the size of the queue constantly grows, this indicates that CPU cores have too many tasks to process, and hence we scale down the number of Lambdas; if the queue quickly shrinks, we scale up the number of Lambdas. The goal here is to stabilize the size of the queue so that the number of Lambdas matches the pace of graph tasks. 



\mysection{Evaluation\label{sec:eval}}

We wrote a total of 11629 SLOC in C++ and CUDA.
There are 10877 of the lines of C++ code: 5393 for graph servers, 2840 for Lambda management (and communication), 1353 for parameter servers, and 1291 for common libraries and utilities.
There are 752 lines of CUDA code for GPU kernels including common graph operations like GCN and mean-aggregators with cuSPARSE~\cite{cusparse}.
Our CUDA code includes deep learning operations such as dense layer and activation layer with cuDNN~\cite{chetlur2014cudnn}. 
\tool supports common stochastic optimizations including \emph{Xavier initialization}~\cite{glorot2010understanding}, \emph{He initialization}~\cite{he-init-iccv15}, a \emph{vanilla SGD optimizer}~\cite{kiefer-sgd-1952am}, and an \emph{Adam optimizer}~\cite{kingma2014adam}, which help training converge smoothly.

\begin{table}[htb]
\vspace{-1em}
    \scriptsize
    \centering
    \begin{tabular}{lrrrr}
        \thead{Graph} & \thead{Size ($|V|$, $|E|$)} & \thead{\# features} & \thead{\# labels} & \thead{Avg. degree} \\
    \hline
        Reddit-small~\cite{graphsage-nips17}  & (232.9K, 114.8M) &  602  & 41 & 492.9 \\
        Reddit-large~\cite{graphsage-nips17} & (1.1M, 1.3B) &  301 & 50 & 645.4 \\
        Amazon~\cite{mcAuley-sigir15, he-www16} & (9.2M, 313.9M)  &  300 & 25 & 35.1 \\
        Friendster~\cite{snap} & (65.6M, 3.6B) & 32 & 50 & 27.5 \\
    \end{tabular}
    \caption{We use 4 graphs, 2 with billions of edges.}
    \label{tab:graphs}
    \vspace{-1em}
\end{table}

\mysubsection{Experiment Setup}
We experimented with four graphs, as shown in Table~\ref{tab:graphs}.
\codeIn{Reddit-small} and \codeIn{Reddit-large} are both generated from the Reddit dataset~\cite{reddit-dataset}.
\codeIn{Amazon} is the \emph{largest graph} in RoC's~\cite{roc-mlsys20} evaluation.
We added a larger 1.8 billion (undirected) edge \codeIn{Friendster} social network graph to our experiments.
For GNN training, we turned undirected edges into two directed edges, effectively doubling the number of edges (which is consistent with how edge numbers are reported in prior GNN work~\cite{roc-mlsys20,neugraph-atc19}).
The first three graphs come with features and labels while \codeIn{Friendster} does not.
For scalability evaluation we generated random features and labels for \codeIn{Friendster}.

We implemented two GNN models on top of \tool: graph convolutional network (GCN)~\cite{kipf-iclr17} and graph attention network (GAT)~\cite{gtn-nips19} with 279 and 324 lines of code. GCN is a popular network that has AV but not AE, while GAT is a recently-developed recurrent network with both AV and AE. Their development is straightforward and other GNN models can be easily implemented on \tool as well. Each model has 2 layers, consistent with those used in prior work~\cite{roc-mlsys20, neugraph-atc19}.

\textit{Value} is the major benefit \tool{} brings to training GNNs.
We define value as a system's \emph{performance per dollar}, computed as $V$ = 1/($T\times C$) where $T$ is the training time and $C$ is the monetary cost. 
For example: if system A trains a network twice as fast as system B, and yet costs the same to train, we say A has twice the value of B.
If one has a time constraint, the most inexpensive option to train a GNN is to pick the system/configuration that meets the time requirement with the best value. 
In particular, value is important for training since users cannot take the cheapest option if it takes too long to train; neither can they take the fastest option if it is extremely expensive in practice.  Throughout the evaluation, we use both value and performance (runtime) as our metrics.

We evaluated several aspects of \tool{}.
\jt{First, we compared several different instance types to determine the configurations
that give us the optimal value for each backend.}
\jt{Second}, we compared several synchronous and asynchronous variants of \tool.
In later subsections, we use our best variant  (which is asynchronous with a
staleness value of 0) in comparisons with other existing systems. 
\jt{Third}, we compared the effects of Lambdas using \tool{} against more traditional
CPU- and GPU-only implementations in terms of value, performance, and scalability. 
Next, we evaluate \tool{} against existing systems.
Finally, we break down our performance and costs to illustrate our system's benefits.

\begin{table}[ht]
  \vspace{-0.5em}
  \small
  \centering
  \begin{tabular}{cccc}
    \thead{Backend} & \thead{Graph} & \thead{Instance Type} & \thead{Relative Value} \\

    \hline

    \multirow{4}{*}{CPU} &
        \multirow{2}{*}{Reddit-large} & r5.2xlarge (4) & 1 \\
            & & c5n.2xlarge (12) & 4.46 \\
        & \multirow{2}{*}{Amazon} & r5.xlarge (4) & 1 \\
            & & c5n.2xlarge (8) & 2.72 \\

    \hline

    \multirow{2}{*}{GPU} &
        \multirow{2}{*}{Amazon} & p2.xlarge (8) & 1 \\
    & & p3.2xlarge (8) & 4.93
  \end{tabular}
  \vspace{-1em}
  \caption{
    \jt{Comparison of the values provided by different instance types. r5 
    and p2 instances provided significantly lower values than the (c5 and p3) instances we chose.}
  }
  \label{tab:instance-selection}
  \vspace{-1em}
\end{table}
\mysubsection{Instance Selection}
\label{sec:instance-selection}
\jt{
To choose the instance types 
for our evaluation, we ran a set of experiments to determine the types that gave
us the best value for each backend.
We compared across memory optimized (r5) and compute optimized (c5) instances,
as well as the p2 and p3 GPU instances, which have K80 and V100 GPUs, respectively.
As r5 offers high memory, we were able to fit the graph in a smaller number of
instances, lowering costs in some cases.
However, due to the smaller amount of computational resources available, training
on the r5 instances typically took nearly $3\times$ as long as computation on c5.
Therefore, as shown in Table~\ref{tab:instance-selection} the average increases
in value c5 instances provided relative to r5 instances are 4.46 and 2.72, respectively, for
\codeIn{Reddit-large} and \codeIn{Amazon}.
We therefore selected c5 as our choice for any CPU based computation.}

\jt{
Similarly, for GPU instances,
training on \codeIn{Amazon} with 8 K80s took 1578 seconds and had a total cost
of \$3.16.
Using 8 V100s took 385 seconds and cost \$2.62\textemdash it improves both
costs and performance, resulting in a value increase of $4.93\times$
compared to training on K80 GPUs.
As value is the main metric which we use to evaluate our system, we
choose the instance type which gives the best value to each different
backend to ensure a fair comparison.}

\jt{
Given these results, we selected the following instances to run our evaluation:
(1) c5, compute-optimized instances, and (2) c5n, compute and network
optimized instances.
c5n instances have more memory and faster networking, but their CPUs have
slightly lower frequency than those in c5.
The base c5 instance has 2 vCPU, 4 GB RAM, and 10 Gbps per-instance network
bandwidth costing \$0.085/h\footnote{These prices are from the Northern
Virginia region.}.
The base c5n instance has 2 vCPU, 5.25 GB RAM (33\% more), and 25 Gbps
per-instance network bandwidth, costing \$0.108/h.
We used the base p3 instance, p3.2xlarge, with Telsa V100 GPUs.
Each p3 base instance has 1 GPU (with 16 GB memory), 8 vCPUs, and 61 GB memory,
costing \$3.06/h.}

\jt{
Each Lambda is a container with 0.11 vCPUs and 192 MB memory.
Lambdas have a static cost of \$0.20 per 1 M requests, and a compute cost
of \$0.01125/h (billed per 100 ms).
This billing granularity enables serverless threads to handle short bursts of massive
parallelism much better than CPU instances. }

\begin{table}[ht]
  \vspace{-0.5em}
  \small
  \centering
  \begin{tabular}{cccc}
    \thead{Model} & \thead{Graph} & \thead{CPU cluster} & \thead{GPU cluster} \\

    \hline

    \multirow{4}{*}{GCN} & Reddit-small & c5.2xlarge (2) & p3.2xlarge (2) \\
    & Reddit-large & c5n.2xlarge (12) & p3.2xlarge (12) \\
    & Amazon & c5n.2xlarge (8) & p3.2xlarge (8) \\
    & Friendster & c5n.4xlarge (32) & p3.2xlarge (32) \\

    \hline

    \multirow{2}{*}{GAT} & Reddit-small & c5.2xlarge (10) & p3.2xlarge (10) \\
    & Amazon & c5n.2xlarge (12) & p3.2xlarge (12) \\
  \end{tabular}
  \vspace{-1em}
  \caption{
    We used (mostly) c5n instances for CPU clusters, and equivalent numbers of p3 instances for GPU clusters.
  }
  \label{tab:cluster-configurations}
  \vspace{-1em}
\end{table}

\jt{
Table~\ref{tab:cluster-configurations} shows our CPU and GPU clusters for each pair of model and graph we evaluated.
For each graph, we picked the number of servers such that they have just enough memory to hold the graph data and their tensors.
For example, \codeIn{Amazon} needs 8 c5n.2xlarge servers (with 16 GB memory) provide enough memory.
For \codeIn{Friendster} we need 32 c5n.4xlarge instances (with a total of 1344 GB memory).
Our goal is to train a model with the minimum amount of resources.
Of course, using more servers will lead to better performance and higher costs (discussed in \S\ref{sec:lambda-effects}).
For all experiments (except \codeIn{Reddit-small}), c5n instances offered the best value.}

\jt{
TPU has become an important type of computation accelerator for machine learning. This paper focuses on AWS and its serverless platform, and hence we did not implement \tool on TPUs. Although we did not compare directly with TPUs, we note several important
features of GNNs that make the limitations of TPUs comparable to GPUs.
First, GNNs are unlike conventional DNNs in that they require large amounts
of data movement for neighborhood aggregation.
As a result, GNN performance is mainly bottlenecked by memory constraints and
the resulting communication overheads (\eg, between GPUs or TPUs), \emph{not} computation
efficiency~\cite{roc-mlsys20}.
Second, GNN training involves computation on large sparse tensors that incur
irregular data accesses, resulting in sub-optimal performance on TPUs which are
optimized for dense matrix operations over regularly structured data.}

\color{black}

\mysubsection{Asynchrony}

\begin{figure*}[t]
\centering
\scriptsize
\begin{tabular}{ccc}
\includegraphics[scale=.3]{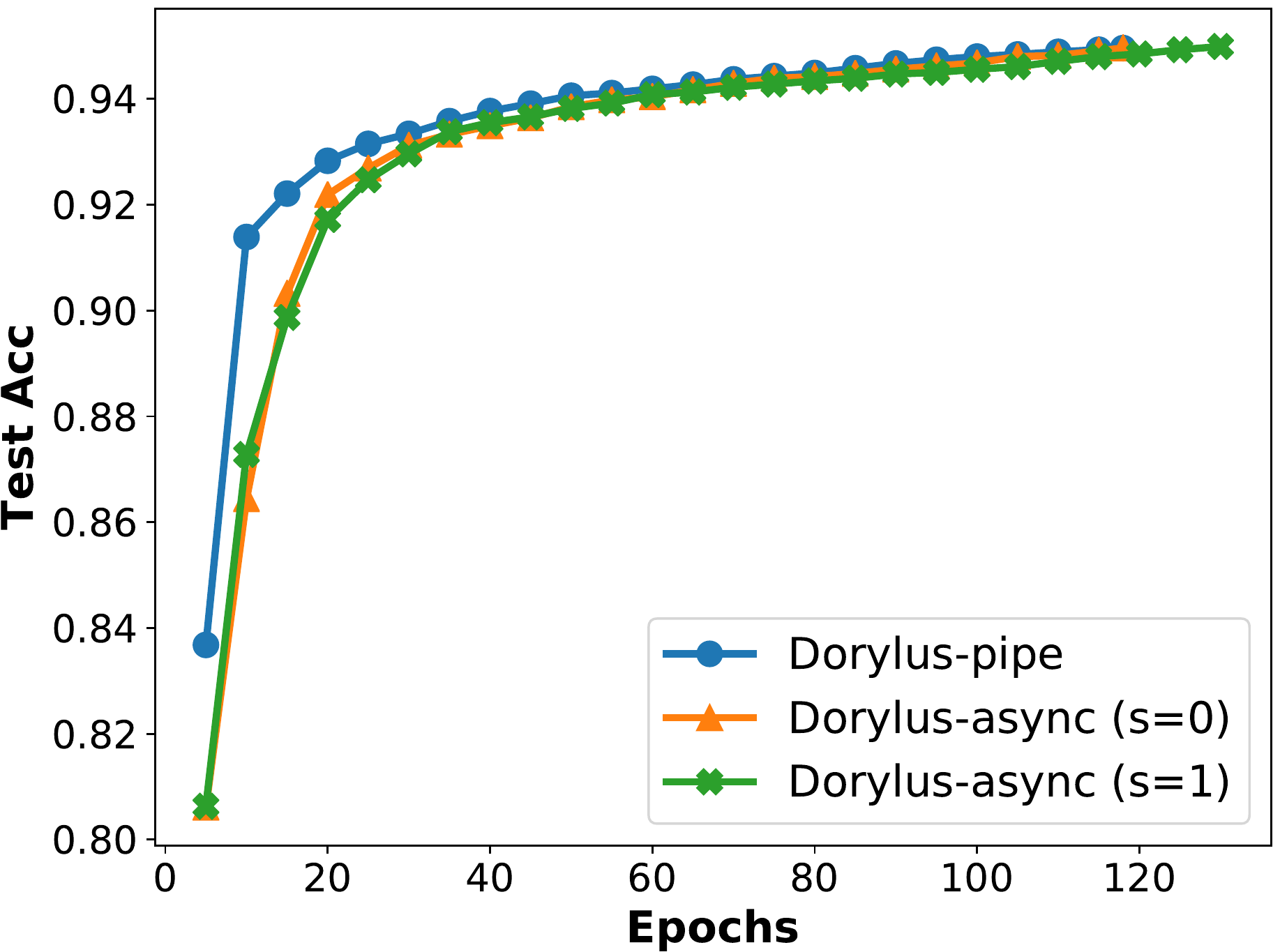}
     &  
\includegraphics[scale=.3]{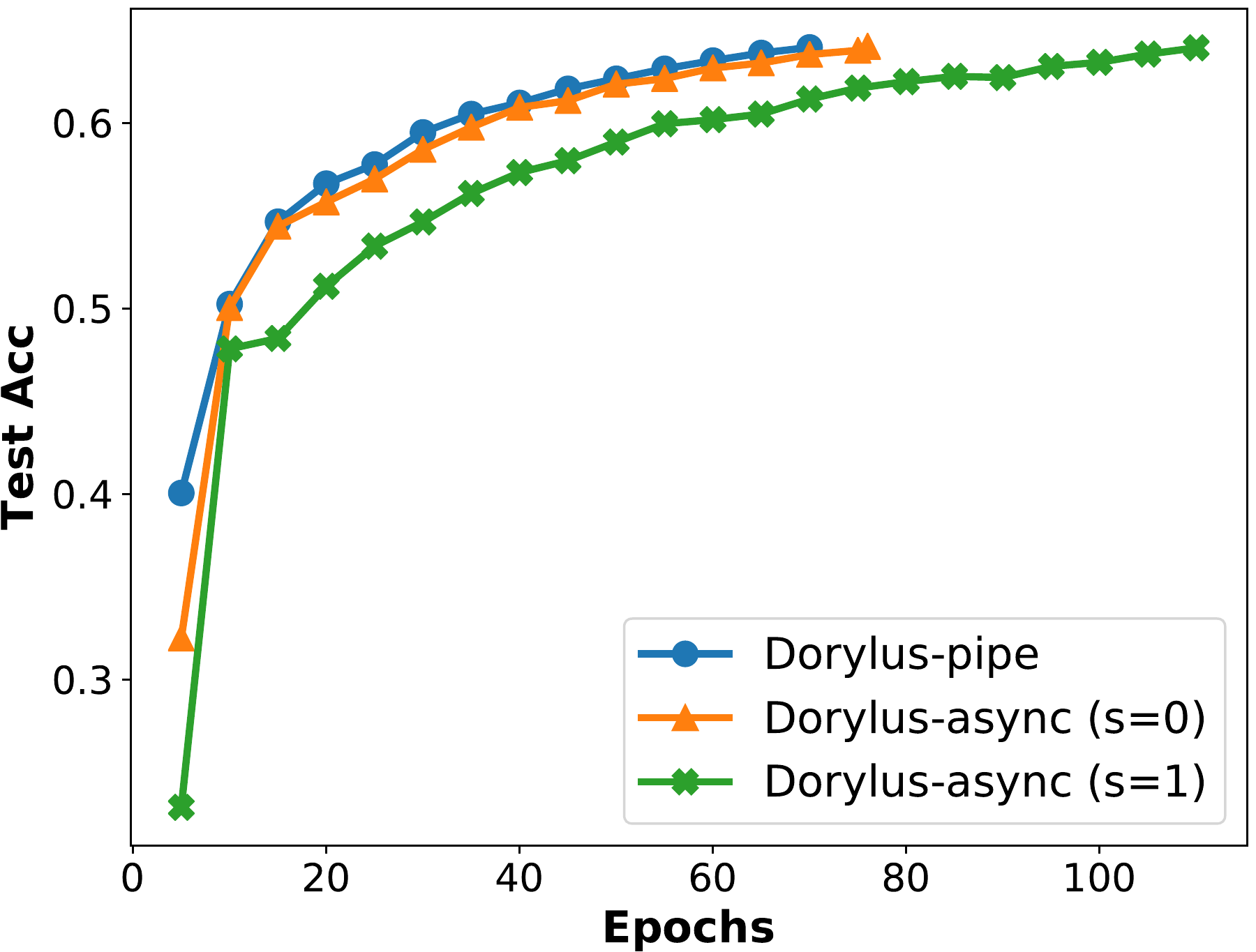}     
    &
\includegraphics[scale=.3]{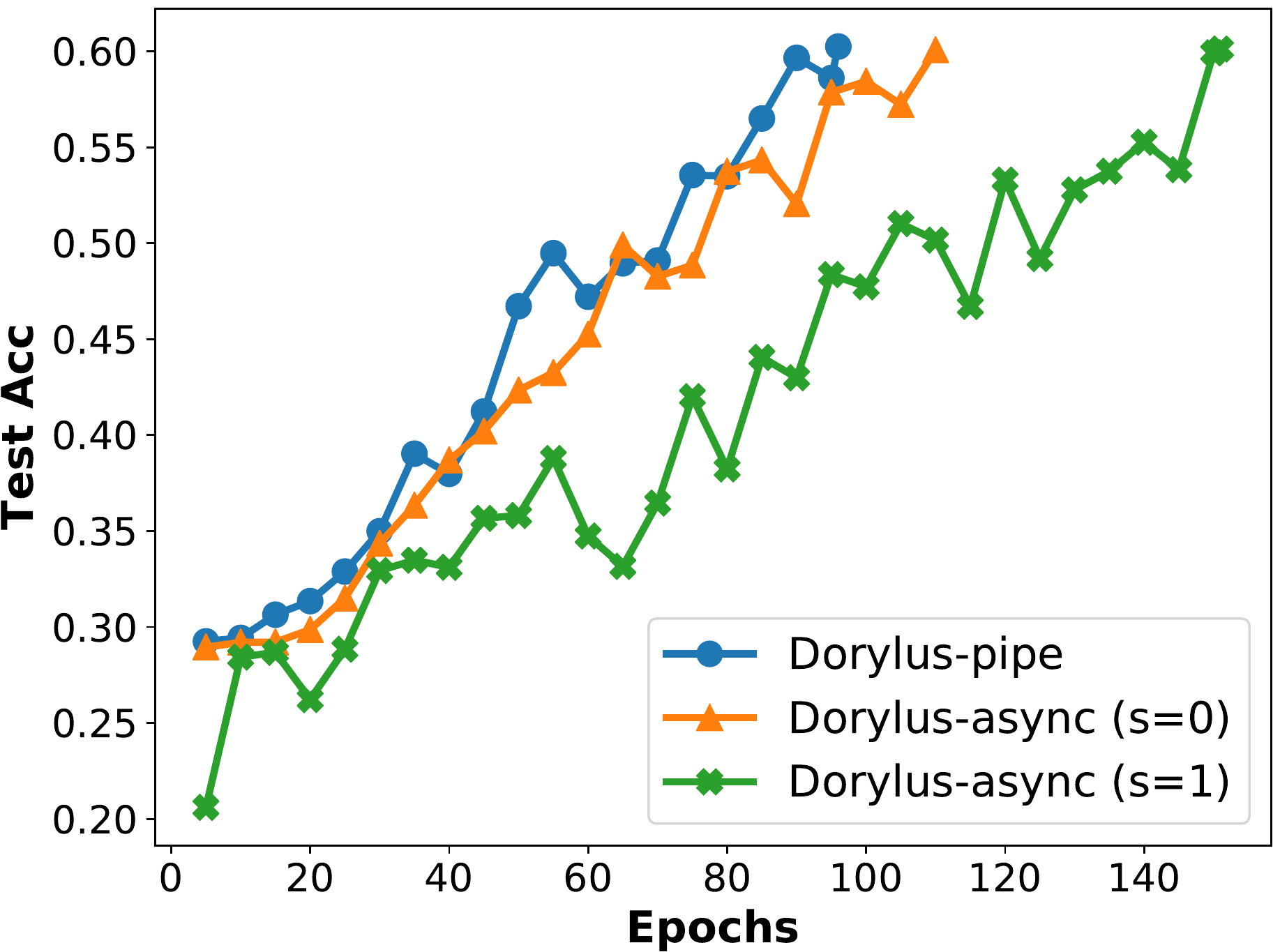}\\
 (a) \codeIn{Reddit-small} &  (b) \codeIn{Amazon} & (c) \codeIn{Reddit-large}  \\
 \textbf{R[$s$=0]:1.00, R[$s$=1]:1.07, Accuracy:94.96\%} & \textbf{R[$s$=0]:1.09, R[$s$=1]:1.57, Accuracy:64.08\% } &  \textbf{R[$s$=0]:1.14, R[$s$=1]:1.58, Accuracy:60.07\%}
 \end{tabular}
 \caption{Asynchronous progress for GCN: All three versions of \tool{} achieve the final accuracy \ie, \textbf{94.96\%}, \textbf{64.08\%}, \textbf{60.07\%} for the three graphs). \codeIn{Friendster} is not included because it does not come with meaningful features and labels. 
 \label{fig:async-progress}}
 \end{figure*}
 
\begin{figure}[!ht]
\vspace{-.5em}
\centering
\begin{tikzpicture}


\begin{axis}[
  axis lines*=left,
  y post scale=0.5,
  enlarge x limits=0.15,
  xtick=data,
  xticklabels={Reddit-small,Reddit-large,Amazon,Friendster},
  xticklabel style={
    font=\small,
    rotate=22.5,
  },
  extra x ticks={4.5},
  extra x tick labels={},
  extra x tick style={grid=major},
  ylabel={Time (relative)},
  ytick={0,1},
  ymin=0,
  ymax=2,
  legend style={
    font=\small,
    draw=none,
    at={(0.5,-0.4)},
    anchor=north,
    legend columns=-1,
  },
  nodes near coords,
  point meta=explicit symbolic,
]
  \addplot+ [
  nodes near coords align={below right},
  ] coordinates {
    (1, 0.9512)
    (2, 0.9718)
    (3, 0.8816)
    (4, 0.6277) [0.63]
  };
  \addplot+ [
   nodes near coords align={above right},
  ] coordinates {
    (1, 0.8315)
    (2, 0.9586)
    (3, 0.8127)
    (4, 0.7188) [0.72]
  };
  \addplot [
    black
  ] coordinates {
    (1, 1)
    (2, 1)
    (3, 1)
    (4, 1)
  };
  \legend{Async (s=0), Async (s=1), Pipe}
\end{axis}
\end{tikzpicture}
\vspace{-1em}
\caption{Per-epoch GCN time for async ($s$=0) and async ($s$=1) normalized to that of pipe. \label{fig:per-epoch}}
\vspace{-1em}
\end{figure}
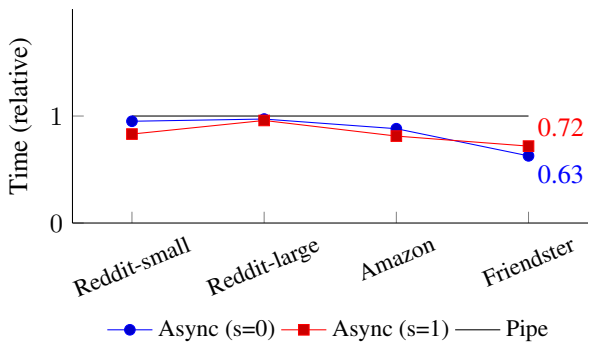

We compare three versions of \tool: a synchronous version with full intra-layer pipelining (pipe), and two asynchronous versions using $s=0$ and $s=1$ as the staleness values over all four graphs. Pipe synchronizes at each \codeIn{Gather} --- a vertex cannot go into the next layer until all its neighbors have their latest values scattered. As a result, all vertex intervals have to be in the same layer in the same epoch. However, inside each layer, pipelining is enabled, and hence different tasks are fully overlapped. Async enables both pipelining and asynchrony (\ie, stashing weights and using stale values at GA). When the staleness value is $s$ = 0, \tool allows a vertex to use a stale value from a neighbor as long as the neighbor is in the same epoch (\eg, can be in a previous layer). In other words, Async ($s$=0) enables fully pipelining across different layers in the same epoch, but pipelining tasks in different epochs are not allowed and synchronization is needed every epoch. Similarly, async ($s$=1) enables a deeper pipeline across two consecutive epochs.

\MyPara{Training Progress.} 
Due to the use of asynchrony, it may take the asynchronous version of \tool more epochs to reach the same accuracy as pipe. To enable a fair comparison, we first ran \tool-pipe until convergence (\ie, the difference of the model accuracy between consecutive epochs is within 0.001, unless otherwise stated) and then used this accuracy as the target accuracy to run async when collecting training time. However, this approach does not work for \codeIn{Friendster}, because it uses randomly generated features/labels and hence accuracy is not a meaningful target. To solve this problem, we computed an average ratio, across the other three graphs, between the numbers of epochs needed for async and pipe, and used this ratio to estimate the training time for \codeIn{Friendster}. For example, this ratio is 1.08 for $s$=0 and 1.41 for $s$=1. As such, we let async ($s$=0) run N$\times$1.08 epochs and async ($s$=1) run N$\times$1.41 epochs when measuring performance for \codeIn{Friendster} where $N$ is the number of epochs pipe runs.

Figure~\ref{fig:async-progress} reports the GCN training progress for each variant, that is, how many epochs it took for a version to reach the target accuracy. Annotated with each figure are two ratios: R[$s$=0] and R[$s$=1], representing the ratio between the number of epochs needed by async ($s$=0/1) and that by \tool-pipe to reach the same target accuracy. On average,  async ($s$=0/1) increases the number of epochs by 8\%/41\%. 

Figure~\ref{fig:per-epoch} compares the per-epoch running time for each version of \tool, normalized to that of pipe. As expected, async has lower per-epoch time; in fact, async ($s$=0) achieves almost the same reduction ($\thicksim$15\%) in per-epoch time as $s$=1. This indicates that choosing a large staleness value has little usefulness --- it cannot further reduce per-epoch time and yet the number of epochs grows significantly. 

To conclude, asynchrony can provide overall performance benefits in general but too large a staleness value leads to slow convergence and poor performance, although the per-epoch time reduces. This explains why async ($s$=0) outperforms async ($s$=1) by a large margin. Overall, async ($s$=0) is \textbf{1.234$\times$} faster than pipe and \textbf{1.233$\times$} than async ($s$=1). It also provides \textbf{1.288$\times$} and \textbf{1.494$\times$} higher value than pipe and async ($s$=1) respectively. Thus we choose it as the default Lambda variant in our following experiments unless otherwise specified.
From this point on, \tool refers to this particular version.

\mysubsection{Effects of Lambdas \label{sec:lambda-effects}}

We developed two traditional variants of \tool{} to isolate the effects of serverless computing using Lambdas, one using CPU-only servers for computations, and the other using GPU-only servers (both without Lambdas).
These variants perform all tensor and graph computations directly on the graph server. They both use \tool{}' (tensor and graph) computation separation for scalability. Note that without computation separation, no existing GPU-based training system has been shown to scale to a billion-edge graph. 

Since Lambdas have weak compute that we cannot find in regular EC2 instances, it is not possible for us to translate Lambda resources directly into equivalent EC2 resources, keeping the total amount of compute constant when selecting the number of servers for each variant. To address this concern, we compared the value of different systems in addition to their absolute times and costs.

\begin{table}[ht]
  \vspace{-0.5em}
  \adjustbox{max width=\linewidth}{
  \centering
  \begin{tabular}{cccrr}
    \thead{Model} & \thead{Graph} & \thead{Mode} & \thead{Time (s)} & \thead{Cost (\$)} \\

    \hline

    \multirow{12}{*}{GCN} & \multirow{3}{*}{Reddit-small}  & Dorylus & 860.6 & 0.20 \\
    & &  CPU only & 1005.4 & 0.19 \\
    & & GPU only & 162.9 & 0.28 \\
    
    \cline{2-5}
    
    & \multirow{3}{*}{Reddit-large} & Dorylus (pipe) & 1020.1 & 1.69 \\
    & & CPU only & 1290.5 & 1.85 \\
    & & GPU only & 324.9 & 3.31 \\

    \cline{2-5}

    & \multirow{3}{*}{Amazon} & Dorylus & 512.7 & 0.79 \\
    & & CPU only & 710.2 & 0.68 \\
    & & GPU only & 385.3 & 2.62 \\
    
    \cline{2-5}
    
    & \multirow{3}{*}{Friendster} & Dorylus & 1133.3 & 13.8 \\
    & & CPU only & 1990.8 & 15.3 \\
    & & GPU only & 1490.4 & 40.5 \\
    
    \hline
    
    \multirow{6}{*}{GAT} & \multirow{3}{*}{Reddit-small} & Dorylus & 496.3 & 1.15 \\
    & & CPU only & 1270.4 & 1.20 \\
    & & GPU only & 130.9 & 1.11 \\
    
    \cline{2-5}
    
    & \multirow{3}{*}{Amazon} & Dorylus & 853.4 & 2.67 \\
    & & CPU only & 2092.7 & 3.01 \\
    & & GPU only & 1039.2 & 10.60 \\
  \end{tabular}
  }
  \caption{
    We ran \tool{} in 3 different modes: ``\tool{}'', our best Lambda variant using async(s=0) (except in one case), the ``CPU only'' variant, and the ``GPU only'' variant.
    For each mode we used multiple combinations of models and graphs.
    For each run we report the total end-to-end running time and the total cost.
  }
  \label{tab:dorylus-raw}
\end{table}

We ran GCN and GAT on our graphs (Table~\ref{tab:dorylus-raw}).
We only ran the GAT model on one small and large graph because it was simply too monetarily expensive (even for our system!).
GAT has an intensive AE computation, which adds cost.
Note that this is \emph{not} a limitation of our system\textemdash our system can scale GAT to graphs larger than \codeIn{Amazon} if cost is not a concern.

Performance and cost by themselves do not properly illustrate the value of \tool{}.
For example, training GAT on \codeIn{Amazon} with \tool{} is both more efficient and cheaper than the CPU- and GPU-only variants. Hence, we report the value results as well. Recall that to compute the value, we take the reciprocal of the total runtime (\ie, the performance or rate of completion) and divide it by the cost. In this case \tool{} with Lambdas provides a $2.75\times$ higher value than CPU-only (\ie, $1/(853.4\times2.67)$ compared to $1/(2092.7\times3.01)$).
Figure~\ref{fig:dorylus-value} shows the value results for all our runs, \emph{normalized to GPU-only servers}.

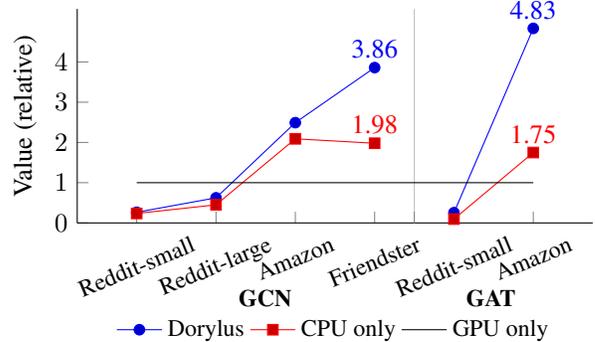
\begin{figure}[ht]
\centering
\begin{tikzpicture}
\begin{axis}[
  axis lines*=left,
  y post scale=0.5,
  enlarge x limits=0.15,
  xtick=data,
  xticklabels={Reddit-small,Reddit-large,Amazon,Friendster,Reddit-small,Amazon},
  xticklabel style={
    font=\small,
    rotate=22.5,
  },
  extra x ticks={4.5},
  extra x tick labels={},
  extra x tick style={grid=major},
  ylabel={Value (relative)},
  ytick={0,1,2,3,4},
  ymin=0,
  legend style={
    font=\small,
    draw=none,
    at={(0.5,-0.4)},
    anchor=north,
    legend columns=-1,
  },
  nodes near coords,
  point meta=explicit symbolic,
]
  \addplot+ [
  ] coordinates {
    (1, 0.265)
    (2, 0.624)
    (3, 2.492)
    (4, 3.860) [3.86]

    (5, 0.255)
    (6, 4.834) [4.83]
  };
  \addplot+ [
  ] coordinates {
    (1, 0.236)
    (2, 0.450)
    (3, 2.090)
    (4, 1.98) [1.98]

    (5, 0.095)
    (6, 1.749) [1.75]
  };
  \addplot [
    black
  ] coordinates {
    (1, 1)
    (2, 1)
    (3, 1)
    (4, 1)
    (5, 1)
    (6, 1)
  };
  \legend{Dorylus,CPU only,GPU only}
\end{axis}
\node at (2.5, -1) {\small \bfseries GCN};
\node at (5.5, -1) {\small \bfseries GAT};
\end{tikzpicture}
\caption{
  \tool{}, with Lambdas, provides up to $2.75\times$ performance-per-dollar than using the CPU-only variant.
}
\label{fig:dorylus-value}
\end{figure}

\tool{} adds value for large, sparse graphs (\ie, \codeIn{Amazon} and \codeIn{Friendster}) for both GCN and GAT, compared to CPU- and GPU-only variants. Sparsity of each graph can be seen from the average vertex degree reported in Table~\ref{tab:graphs}. As shown, \codeIn{Amazon} and \codeIn{Friendster} are much more sparse than \codeIn{Reddit-small} and \codeIn{Reddit-large}. 
For these graphs, the GPU-only variant has the lowest value, even compared to the CPU-only variant.
In most cases, the CPU-only variant provides twice as much value (\ie, performance per dollar) than the GPU-only variant.
\tool{} adds another leap in value over the CPU-only variant.

However, for small dense graphs (\ie, \codeIn{Reddit-small} and \codeIn{Reddit-large}), both \tool and the CPU-only variant have a value lower that that of the GPU-only variant (\ie, below 1 in Figure~\ref{fig:dorylus-value}). \tool{} always provides more value than the CPU-only variant. These results suggest that GPUs may be better suited to process small, dense graphs rather than large, sparse graphs. 



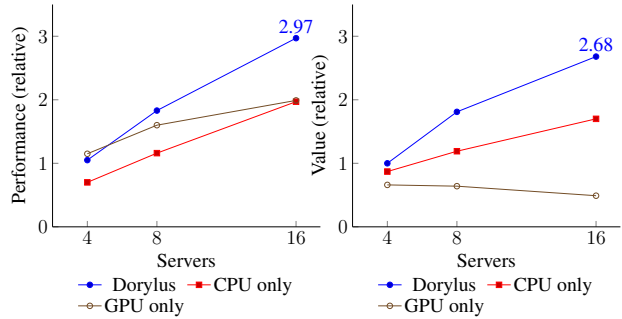
\begin{figure}[!ht]
    \centering
    \vspace{-1em}
  \adjustbox{max width=\linewidth}{
    \begin{tabular}{cc}
        \hspace{-1em}
    \resizebox{0.25\textwidth}{!}{\begin{tikzpicture}
\begin{axis}[
  axis lines*=left,
  xtick=data,
  xmin=2,
  xlabel=Servers,
xlabel style={
    font=\Large,
  },
  xticklabel style={
    font=\Large,
  },
  ylabel={Performance (relative)},
  ylabel style={
    font=\Large,
  },
  ymax=3.5,
  ytick={0,1,2,3},
  yticklabel style={
    font=\Large,
  },
  ymin=0,
  legend style={
    font=\Large,
    draw=none,
    at={(0.5,-0.2)},
    anchor=north,
    legend columns=2,
  },
  nodes near coords,
  nodes near coords style={font=\Large}, 
  point meta=explicit symbolic,
]
  \addplot+ [
  ] coordinates {
    (4, 1.05)
    (8, 1.83)
    (16, 2.97) [2.97]
  };
  \addplot+ [
  ] coordinates {
    (4, 0.7)
    (8, 1.16)
    (16, 1.97)
  };
  \addplot+ [
    mark=o,
  ] coordinates {
    (4, 1.15)
    (8, 1.6)
    (16, 1.99)
  };
\addlegendentry{Dorylus}
\addlegendentry{CPU only}
\addlegendentry{GPU only}
\end{axis}
\end{tikzpicture}}
    &
    \hspace{-2em}
    \resizebox{0.25\textwidth}{!}{\begin{tikzpicture}
\begin{axis}[
  axis lines*=left,
  xtick=data,
  xmin=2,
  xlabel={Servers},
  xlabel style={
    font=\Large,
  },
  xticklabel style={
    font=\Large,
  },
  ylabel={Value (relative)},
  ylabel style={
    font=\Large,
  },
  ytick={0,1,2,3},
  yticklabel style={
    font=\Large,
  },
  ymin=0,
  ymax=3.5,
  legend style={
    font=\Large,
    draw=none,
    at={(0.5,-0.2)},
    anchor=north,
    legend columns=2,
  },
  nodes near coords,
    nodes near coords style={font=\Large}, 
  point meta=explicit symbolic,
]
  \addplot+ [
  ] coordinates {
    (4, 1.0)
    (8, 1.81)
    (16, 2.68) [2.68]
  };
  \addplot+ [
  ] coordinates {
    (4, 0.87)
    (8, 1.19)
    (16, 1.7)
  };
  \addplot+ [
    mark=o,
  ] coordinates {
    (4, 0.66)
    (8, 0.64)
    (16, 0.49)
  };
  \legend{Dorylus,CPU only,GPU only}
\end{axis}
\end{tikzpicture}}\\
    \end{tabular}
    }
    \caption{Normalized GCN training performance and value over \codeIn{Amazon} with varying numbers of graph servers. \label{fig:scale_out}}
\end{figure}

\MyPara{Scaling Out.}
\tool{} can gain even more value by scaling out to more servers, due to the burst parallelism provided by Lambdas and deep pipelining.
To understand the impact of the number of servers on performance/costs, we varied the number of GSes when training a GCN over \codeIn{Amazon}.
In particular, we ran \tool{} and the CPU-only variant with 4, 8, and 16 c5n.4xlarge servers, and the GPU-only variant with the same numbers of p3.xlarge servers.
Figure~\ref{fig:scale_out} reports their performance and values, normalized to those of \tool{} under 4 servers. 

In general, \tool scales well in terms of both performance and value.
\tool gains a 2.82$\times$ speedup with only 5\% more cost when the number of servers increases from 4 to 16, leading to a 2.68$\times$ gain in its value.
As shown in Figure~\ref{fig:scale_out}(b), \tool's value curve is always above that of the CPU-only variant. Furthermore, \emph{\tool can roughly provide the same value as the CPU-only variant with only half of the number of servers}.
For example, \tool with 4 servers provides a comparable value to the CPU-only variant with 8 servers; \tool with 8 servers provides more value to the CPU-only variant with 16 servers. These results suggest that as more servers are added, the value provided by \tool{} increases, at a rate much higher than the value increase of the CPU-only variant. As such, \tool is always a better choice than the CPU-only variant under the same monetary budget.

\MyPara{Other Observations.} In addition to the results discussed above, we make three other observations on performance.

\underline{Our first observation} is that the more sparse the graph, the more useful \tool is. For \codeIn{Amazon} and \codeIn{Friendster}, \tool even outperforms the GPU-only version for two reasons:

First, for all the three variants, the fraction of time on \codeIn{Scatter} is significantly larger when training over \codeIn{Friendster} and \codeIn{Amazon} than \codeIn{Reddit-small} and \codeIn{Reddit-large}. This is, at first sight, counter-intuitive because one would naturally expect less efforts on inter-partition communications for sparse graphs than dense graphs. A thorough inspection discovered that the \codeIn{Scatter} time actually depends on a \emph{combination} of the number of ghost vertices and inter-partition edges. For the two Reddit graphs, they have many inter-partition edges, but very few ghost vertices, because (1) their $|V|$ is small and (2) many inter-partition edges come from/go to the same ghost vertices due to their high vertex degrees. 

Second, \codeIn{Scatter} takes much longer time in GPU clusters. Moving ghost data between GPU memories on different nodes is much slower than data transferring between CPU memories. 
As a result, the poor performance of the GPU-only variant is due to a
combinatorial effect of these two factors: \tool scatters significantly
more data for \codeIn{Friendster} and \codeIn{Amazon}, which amplifies
the negative impact of poor scatter performance in a GPU cluster.
\jt{Note that p3 also offers multi-GPU servers, which may potentially reduce scatter time. 
We have also experimented with these servers, but we still observed long scatter time due to extensive communication between between servers and GPUs.
Reducing such communication costs requires fundamentally different techniques such as those proposed by NeuGraph~\cite{neugraph-atc19}. We leave the incorporation of such techniques to future work. 
}

\underline{Our second observation} is that Lambda threads are more effective in boosting performance for GAT than GCN. This is because GAT includes an additional AE task, which performs intensive per-edge tensor computation and thus benefits significantly from a high degree of parallelism. 

\underline{Our third observation} is that \tool achieves comparable performance with the CPU-only variant that uses twice as many servers. For example, the training time of \tool  under 4 servers is only 1.1$\times$ longer than that of the CPU only variant with 8 servers. Similarly, \tool under 8 servers is only 1.05$\times$ slower than the CPU only variant with 16 servers. These results demonstrate our efficient use of Lambdas.

\mysubsection{Comparisons with Existing Systems \label{sec:comparison}}

Our goal was to compare \tool with all existing GNN tools. However, NeuGraph~\cite{neugraph-atc19} and AGL~\cite{agl-vldb20} are not publicly available; neither did their authors respond to our requests.  Roc~\cite{roc-mlsys20} is available but we could not run it in our environment due to various CUDA errors; we were not able to resolve these errors after multiple email exchanges with the authors. Roc was not built for scalability because each server needs to load the entire graph into its memory during processing. This is not possible when processing billion-edge graphs. 
This subsection focuses on the comparison of \tool, DGL~\cite{dgl}, which is a popular GNN library with support for sampling, as well as AliGraph~\cite{aligraph-kdd19}, which is also a sampling-based system that trains GNNs only with CPU servers.
All experiments use the cluster configuration specified above for each graph unless otherwise stated.

DGL represents an input graph as a (sparse) matrix; both graph and tensor computations are executed by PyTorch or MXNet as matrix multiplications. We experimented with two versions of DGL, one with sampling and one without. DGL-non-sampling does full-graph training on a single machine. DGL-sampling partitions the graph and distributes partitions to different machines. Each machine performs sampling on its partition and trains a GNN on sampled subgraphs.

AliGraph runs in a distributed setting with a server that stores the graph information.
A set of clients query the server to obtain graph samples and use them as minibatches
for training. 
Similar to DGL, AliGraph uses a traditional ML framework as a backend and performs all of its
computation as tensor operations.

\begin{figure}[!ht]
\centering
\scriptsize
\begin{tabular}{cc}
\hspace{-1em}
\begin{minipage}[t]{.48\linewidth}
\includegraphics[scale=.25]{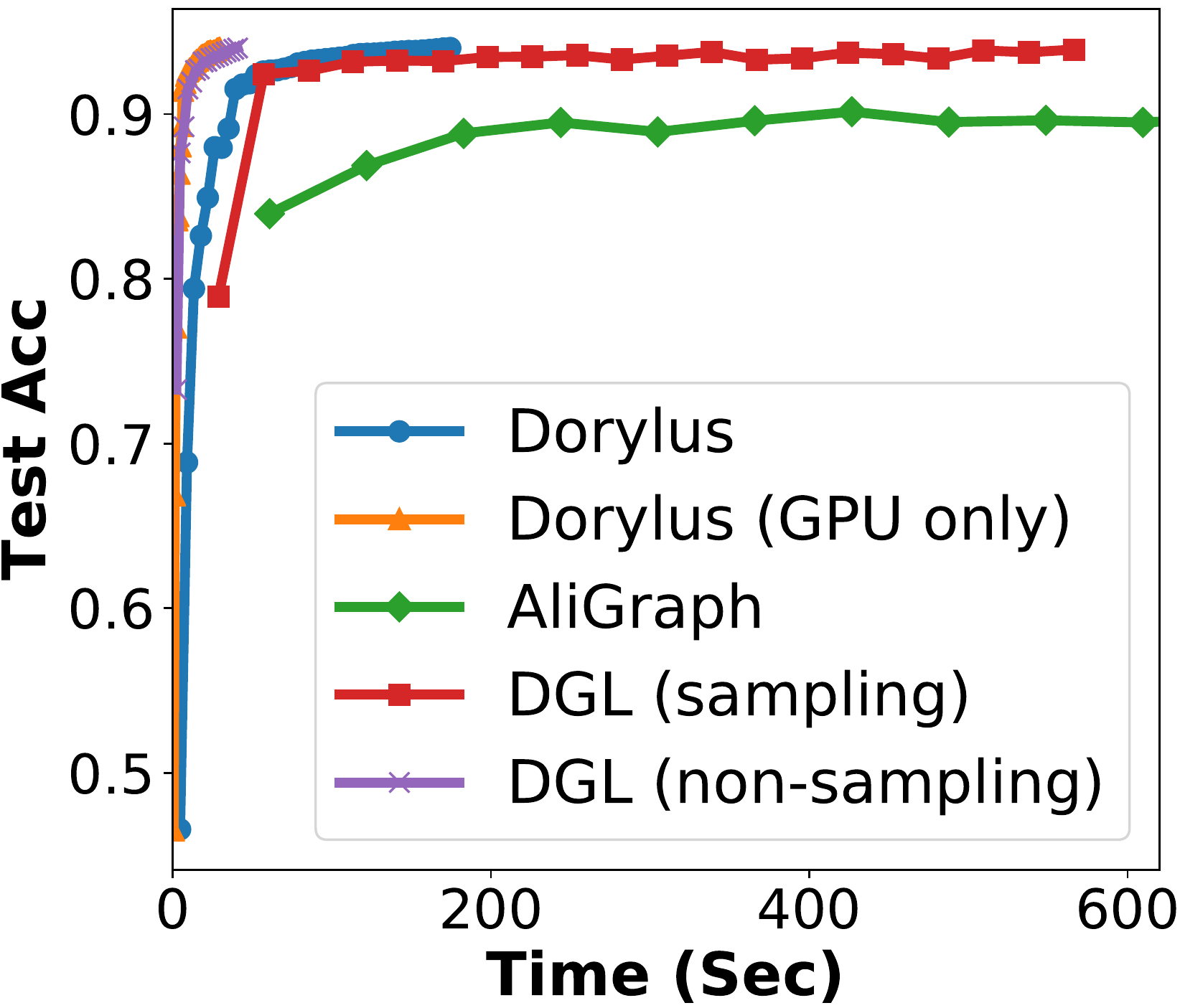}
\end{minipage}
&
\hspace{-1.5em}
\begin{minipage}[t]{.48\linewidth}
\includegraphics[scale=.25]{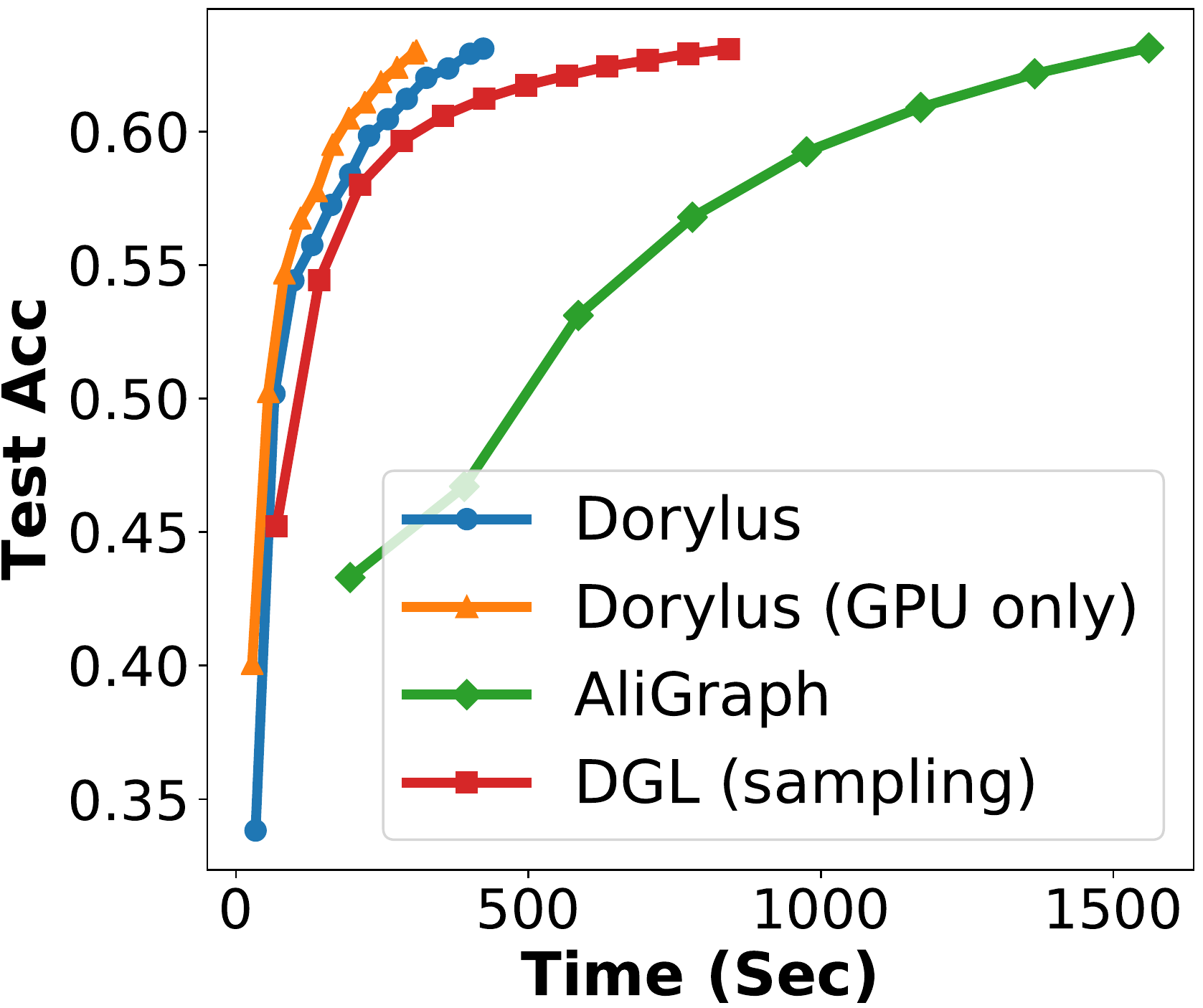}    
\end{minipage}\\
(a)\codeIn{Reddit-small}
&
(b) \codeIn{Amazon}\\
\end{tabular}
 \caption{Accuracy comparisons between \tool, \tool (GPU only), AliGraph, DGL (sampling), and DGL (non-sampling). DGL (non-sampling) uses a single V100 GPU and could not scale to \codeIn{Amazon}.
 Each dot indicates five epochs for \tool and DGL (non-sampling), and one epoch for DGL (sampling) and AliGraph. \label{fig:progress}}
 \end{figure}

\MyPara{Accuracy Comparison with Sampling.} Figure~\ref{fig:progress} reports the
accuracy-time curve for five configurations: \tool, \tool (GPU-only), DGL (sampling),
DGL (non-sampling), and AliGraph, over \codeIn{Reddit-small} and \codeIn{Amazon}. 
When run enough epochs to fully converge, \tool can reach an accuracy of
\textbf{95.44\%} and \textbf{67.01\%}, respectively, for the two graphs.
DGL (non-sampling) can run only on the \codeIn{Reddit-small} graph, reaching 94.01\%
as the highest accuracy.
DGL (sampling) is able to scale to both graphs, and its accuracy reaches 93.90\%
and 65.78\%, respectively, for \codeIn{Reddit-small} and \codeIn{Amazon}. 
AliGraph is able to scale to both \codeIn{Reddit-small} and \codeIn{Amazon}.
On \codeIn{Reddit-small} it reaches a maximum accuracy of 91.12\% and 65.23\% on
\codeIn{Amazon}.

%
%
%
%
\begin{table}[t]
  \small
  \centering
  \begin{tabular}{cccc}
    \thead{Graph} & \thead{System} & \thead{Time (s)} & \thead{Cost (\$)} \\
    \hline
    
    \multirow{5}{*}{Reddit-small} & Dorylus & 165.77 & 0.045 \\
    & Dorylus (GPU only) & 28.06 & 0.052 \\
    & DGL (sampling) & 566.33 & 0.480 \\
    & DGL (non-sampling) & 33.64 & 0.028 \\
    & AliGraph & -- & -- \\
    \hline
    
    \multirow{5}{*}{Amazon} & Dorylus & 415.23 & 0.654 \\
    & Dorylus (GPU only) & 308.27 & 2.096 \\ 
    & DGL (sampling) & 842.49 & 5.728 \\
    & DGL (non-sampling) & -- & -- \\
    & AliGraph & 1560.66 & 1.498 \\
  \end{tabular}
  \vspace{-1em}
  \caption{
    Evaluation of end-to-end performance and costs of \tool and other GNN training systems.
    Each time reported is the time to reach the target accuracy.
  }
  \label{tab:existing-systems}
\end{table}

\MyPara{Performance.} To enable meaningful performance comparisons and make training finish in
a reasonable amount of time, we set 93.90\% and 63.00\% as our target accuracy for
the two graphs.
As shown in Figure~\ref{fig:progress}(a), \tool (GPU only) has the best performance, followed by DGL (non-sampling).
Since \codeIn{Reddit-small} is a small graph that fits into the memory of a single (V100) GPU,
DGL (non-sampling) performs much better than DGL (sampling), which incurs \emph{per-epoch} sampling overheads.
To reach the same accuracy (93.90\%), \tool is 3.25$\times$ faster than DGL (sampling), but 5.9$\times$ slower than \tool (GPU only).
AliGraph is unable to reach our target accuracy after many epochs.

For the \codeIn{Amazon} graph, DGL cannot scale without sampling.
As shown in Figure~\ref{fig:progress}(b), to reach the same target accuracy,
\tool is 1.99$\times$ faster than DGL (sampling), and 1.37$\times$ slower than \tool (GPU only). 
AliGraph is able to reach the target accuracy for \codeIn{Amazon}. However, \tool is significantly faster.
As these results show, graph sampling improves scalability at the cost of increased overheads and reduced accuracy. 

The times reported for \tool and its GPU-only variant in Table~\ref{tab:existing-systems} are smaller than those reported in Table~\ref{tab:dorylus-raw}. This is due to the lower target accuracy we set for these experiments.

\MyPara{Value Comparison.} To demonstrate the promise of Dorylus, we compared these systems
using the value metric.
As expected, given the small size of the \codeIn{Reddit-small} graph, the GPU-based systems perform quite well.
In fact, in this case the normalized value of DGL (non-sampling) is 1.48, providing a higher value than \tool (GPU only).
However, as mentioned earlier, DGL cannot scale without sampling; hence, this benefit is limited only to small graphs.
As we process \codeIn{Amazon}, the value of \tool quickly improves as is consistent with our findings earlier (on large, sparse graphs).
With this dataset, \tool provides a higher performance-per-dollar rate than \emph{all} the other systems\textemdash
17.7$\times$ the value of DGL (sampling) and 8.6$\times$ the value of AliGraph.


\mysubsection{Breakdown of Performance and Costs\label{sec:breakdown}}

\begin{figure}[!ht]
    \centering
    \scriptsize
    \begin{tabular}{cc}
        \hspace{-1.2em}
    \includegraphics[scale=.3]{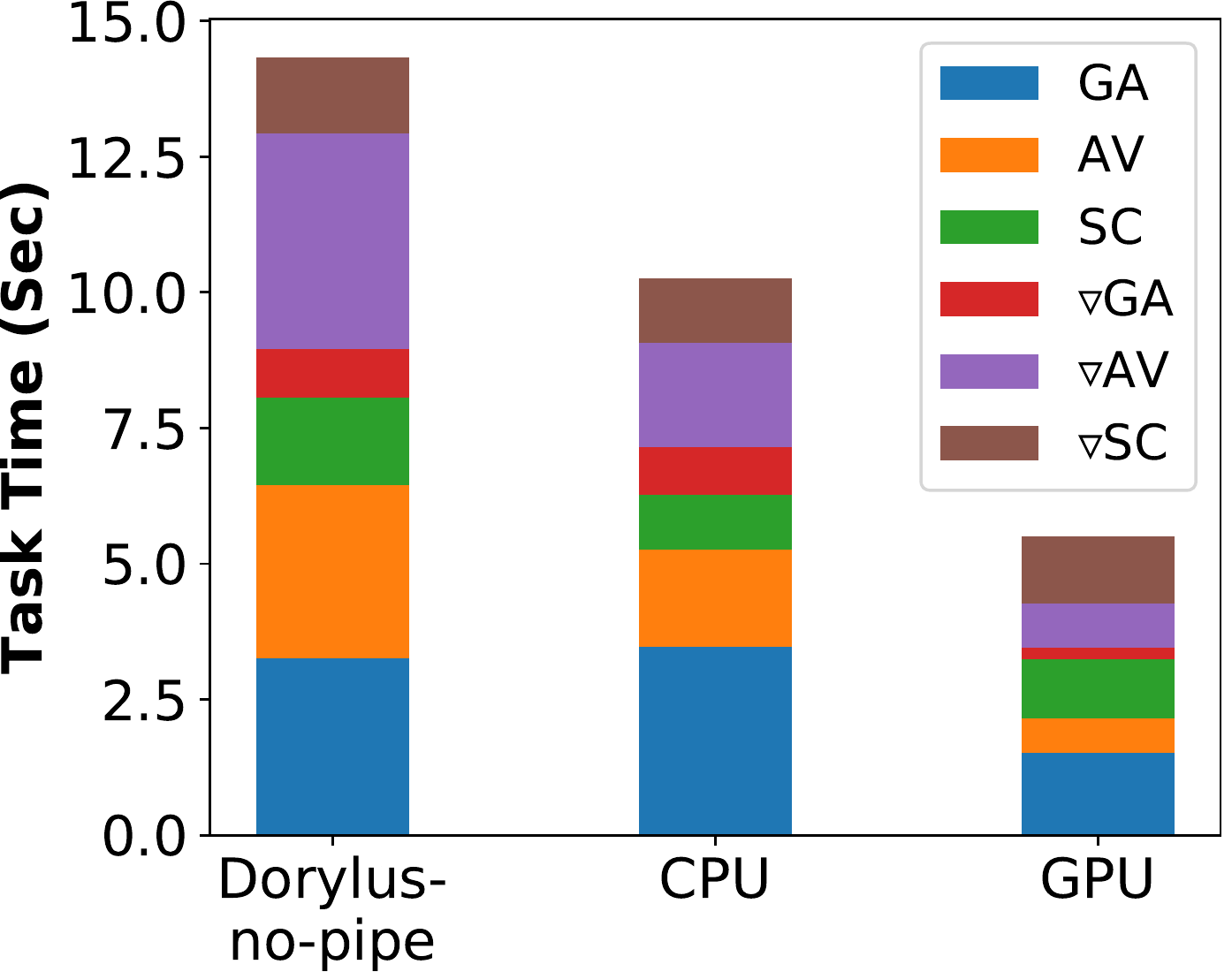}
    &
    \hspace{-1em}
    \includegraphics[scale=.3]{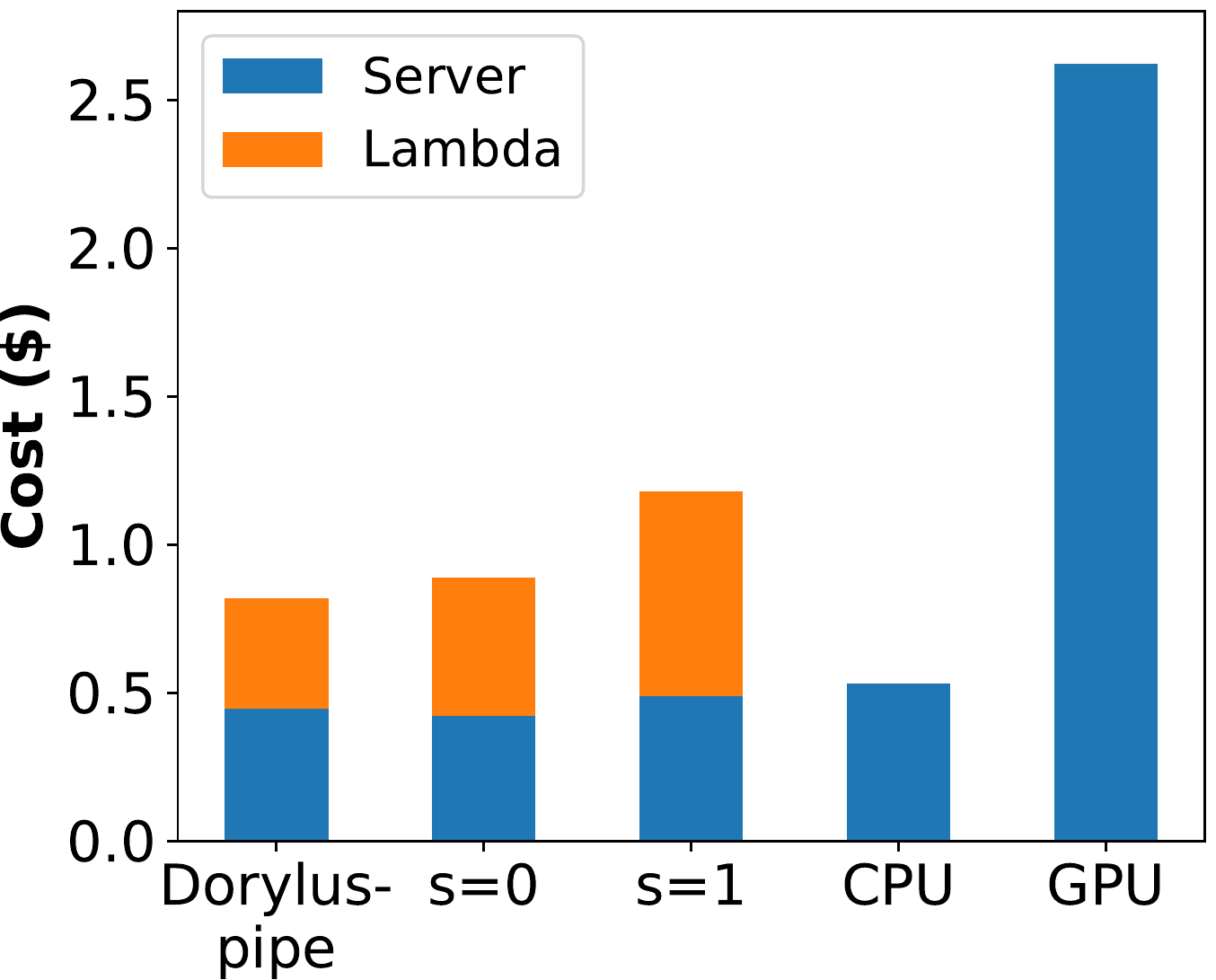}\\
    (a) Task time breakdown & (b) Cost breakdown \\
    \end{tabular}
    \caption{Time and cost breakdown for the \codeIn{Amazon} graph.\label{fig:breakdown_amazon}}
\end{figure}

Figure~\ref{fig:breakdown_amazon} shows a breakdown in task time (a) and costs (b) for training a GCN over the \codeIn{Amazon} graph. In Figure~\ref{fig:breakdown_amazon}(a), to understand the time each task spends, we disabled pipelining and asynchrony in \tool, producing a version referred to as no-pipe, in which different tasks never overlap. This makes it possible for us to collect each task's running time.  Note that no-pipe represents a version that uses Lambdas \naively to train a DNN. Without pipelining and overlapping Lambdas with CPU-based tasks, 
we saw a \textbf{1.9$\times$} degradation, making no-pipe lose to both CPU and GPU in training time. 

As shown, the tasks GA, AV, and $\triangledown$AV take the majority of the time. 
Another observation is that to execute the tensor computation AV, GPU is the most efficient backend and Lambda is the least efficient one. This is expected --- Lambdas have less powerful compute (much less than CPUs in the c5 family) and high communication overheads.  Nevertheless, these results also demonstrate that {\em when CPUs on graph servers are fully saturated with the graph computation, large gains can be obtained by running tensor computation in Lambdas that fully overlap with CPU tasks!} 

To compute the cost breakdown in Figure~\ref{fig:breakdown_amazon}(b), we simply calculated the total amounts of time for Lambdas and GSes for each of the five \tool variants and used these times to compute the costs of Lambdas and servers. Due to \tool' effective use of Lambdas, we were able to run a large number of Lambdas for the forward and backward pass. As such, the cost of Lambdas is about the same as the cost of CPU servers. 

\mysection{Related Work} 
\tool is the first system that successfully uses tiny Lambda threads to train a GNN by exploiting various graph-related optimizations. 
There are three categories of techniques in parallelization (\S\ref{sec:parallelism}), GNN training (\S\ref{sec:gnn-training}), and graph systems (\S\ref{sec:graph}). 
\mysubsection{Parallel Computation for Model Training\label{sec:parallelism}} 
How to exploit parallelism in model training is a topic that has been extensively studied. 
There are two major dimensions in how to effectively parallelize the training work: (1) what to partition and (2) how to synchronize between workers.

\MyPara{What to Partition.} The most straightforward parallelism model is \emph{data parallelism}~\cite{all-reduce-16, dawnbench-sigops19,geeps-eurosys16,minibatch-sgd-17,seide2014-bit-interspeech14,speech-dnns-icassp14,communication-mpich-jhpca05,poseidon-atc17}, where \emph{inputs} are partitioned and processed by individual workers. Each worker learns parameters (weights) from its own portion of inputs and periodically shares its parameters with other workers to obtain a global view. Both share-memory systems~\cite{all-reduce-16, speech-dnns-icassp14, minibatch-sgd-17} and distributed systems~\cite{parameterserver, poseidon-atc17, adam} have been developed for data-parallel training.
Another parallelization strategy is to partition the work, often referred to as \emph{model parallelism}~\cite{device-placement-icml17} where the operators in a model are partitioned and each worker evaluates and updates only a subset of parameters \wrt its model partition for all inputs.

A recent line of work develops techniques for \emph{hybrid parallelism}~\cite{pipedream-sosp19,gpipe-18,flexflow-mlsys18, parallelize-cnn14}. PipeDream~\cite{pipedream-sosp19} adds pipelining into model parallelism to fully utilize compute without introducing significant stalls. Although \tool also uses pipelining, tasks on a \tool pipeline are much finer-grained. For example, instead of splitting a model into layers, we construct graph and tensor tasks in such a way that graph tasks can be parallelized on graph servers, while each tensor task is small enough to fit into a Lambda's resource profile. \tool uses pipelining to overlap graph and tensor computations specifically to mitigate Lambdas' network latency. FlexFlow~\cite{flexflow-mlsys18} automatically splits an iteration along four dimensions.

\MyPara{How Workers Synchronize.}  When workers work on different portions of inputs (\ie, data parallelism), they need to share their learned parameters with other workers. Parameter updating requires synchronization between workers. For share-memory systems, they often rely on primitives such as \codeIn{all\_reduce}~\cite{all-reduce-16} that broadcasts each worker's parameters to all other workers. Distributed systems including \tool use parameter servers~\cite{parameterserver, poseidon-atc17, adam}, which periodically communicate with workers for updating parameters. The most commonly-used approach for synchronization is the bulk synchronous parallel (BSP) model, which poses a barrier at the end of each epoch. All workers need to wait for gradients from other workers at the barrier. Wait-free backpropagation~\cite{poseidon-atc17} is an optimization of the BSP model.

Since synchronous training often introduces computation stalls, \emph{asynchronous} training~\cite{all-reduce-16,geeps-eurosys16} has been proposed to reduce such stalls --- each worker proceeds with the next input minibatch before receiving the gradients from the previous epoch. An asynchronous approach reduces time needed for each epoch at the cost of increased epochs to reach particular target accuracy. This is because allowing workers to use parameters learned in epoch $m$ to perform forward computations in epoch $n$ ($n \neq m$) leads to \emph{statistical inefficiency}.
This problem can be mitigated with a hybrid approach such as \emph{bounded staleness}~\cite{pipedream-sosp19, cui-atc14, aspire, hogwild-nips11}.

\mysubsection{GNN Training and Graph Systems\label{sec:gnn-training}}
As the GNN family keeps growing~\cite{gnn-recommendation-aaai19, gtn-nips19, cnn-graphs-nips16, gnn-motif-attention-cikm19, gnn-recommender-kdd18, gnn-review-18, ggsnn-iclr16,jia-kdd20, aligraph-kdd19,agl-vldb20}, developing efficient and scalable GNN training systems becomes popular. GraphSage~\cite{graphsage-nips17} uses graph sampling, NeuGraph~\cite{neugraph-atc19} extends GNN training to multiple GPUs, and RoC~\cite{roc-mlsys20} uses dynamic graph partitioning to achieve efficiency. Other systems that can scale to large graphs are all based on sampling~\cite{aligraph-kdd19,agl-vldb20}.

Programming frameworks such as DGL~\cite{dgl} have been proposed to create a graph-parallel interface (\ie, GAS) for developers to easily mix graph operations with NNs. However, such frameworks still represent the graph as a matrix and push it to an underlying training framework such as TensorFlow for training. We solve this fundamental scalability problem with a ground-up system redesign that separates the graph computation from the tensor computation. 

\mysubsection{Graph-Parallel Systems\label{sec:graph}}
There exists a body of work on scalable and efficient graph systems of many kinds: single-machine share-memory systems~\cite{ligra,galois,grace,mariappan-eurosys19,mariappan-eurosys21}, disk-based out-of-core systems~\cite{graphchi,xstream,gridgraph,graphq,mmap,flashgraph,graspan,turbograph,vora-atc16, maass-eurosys17, lumos-atc19, ai-atc17,rstream}, and distributed systems~\cite{pregel,graphlab, powergraph, cyclops, roy-sosp15, powerlyra, zhu-osdi16, zhang-osdi16, shi-osdi16,   kickstarter, naiad, wu-socc15, bu-vldb14, vora-taco16,vora-asplos17-coral}. 
These systems were built on top of a graph-parallel computation model, whether it is vertex-centric or edge-centric. Inspired by these systems, \tool formulates operations involving the graph structure as graph-parallel computation and runs it on CPU servers for scalability.

\mysection{Conclusion} 

\tool is a distributed GNN training system that scales to large billion-edge graphs with low-cost cloud resources.
We found that CPU servers, in general, offer more performance per dollar than GPU servers for large sparse graphs.
Adding Lambdas added $2.75\times$ more performance-per-dollar than CPU only servers, and $4.83\times$ more than GPU only servers.
Compared to existing sampling-based systems \tool{} is up to $3.8\times$ faster and $10.7\times$ cheaper.
Based on the trends we observed \tool{} can scale to even larger graphs than we evaluated, offering even higher values.

\section*{Acknowledgments}
We thank the anonymous reviewers for
their comments. We are 
grateful to our shepherd Amar Phanishayee for his feedback.
This work is supported by NSF grants CCF-1629126, CNS-1703598, CCF-1723773, CNS-1763172, CCF-1764077,  CNS-1907352, CNS-1901510, CNS-1943621, CHS-1956322, CNS-2007737, CNS-2006437, CNS-2106838, ONR grants N00014-16-1-2913 and N00014-18-1-2037, as well as a Sloan Fellowship. 
\appendix
\mysection{Artifact Appendix}
\mysubsection{Artifact Summary}
\tool is a distributed GNN training system that scales to large billion-edge 
graphs using cheap cloud resources\textemdash specifically CPU servers and serverless
threads.
It launches a set of graph servers which are used for processing graph data
and doing operations such as gather and scatter.
In addition, parameter servers hold the weights for the model.
It can be configured to run with multiple different backends, such as a pure
CPU backend and a GPU backend.
By separating the graph and tensor components of a graph neural network
\tool is able to effectively utilize serverless threads by providing a
deep asynchronous-parallel pipeline in which tensor and graph operations
are overlapped.
By doing this \tool significantly improves the performance-per-dollar of
serverless training over both the CPU and GPU backends.


\mysubsection{Artifact Check-list}
{\small
\squishlist
  \item {\bf Hardware: AWS cloud account}
  \item {\bf Public link: \url{https://github.com/uclasystem/dorylus} }
  \item {\bf Code licenses: \textbf{The GNU General Public License (GPL)} }
\squishend
}

\mysubsection{Description}
\mysubsubsection{\tool's Codebase} 
\tool contains the following three components: 
\squishlist
\item The Graph Server which performs graph operations and manages
    Lambda threads (which can also use CPU and GPU backends)
\item The Weight Server which holds the model parameters and sends them
    to the workers
\item The Lambda functions which can be uploaded to AWS to be used during
    training
\squishend

\mysubsubsection{Deploying \tool}  

To build \tool, the first step is to make sure you have the following dependencies
installed on your local machine:
\squishlist
    \item \codeIn{awscli}
    \item \codeIn{python3-venv}
\squishend

Make sure to run \codeIn{aws configure} and set up your credentials to allow
you to have access to AWS services.
Once these are installed, we need to download the code and setup the
environment:

\begin{mdframed}[style=MyFrame,nobreak=true]
\codeIn{git clone} \\ \codeIn{git@github.com:uclasystem/dorylus.git} \\
\\
\codeIn{cd dorylus/} \\
\codeIn{git checkout v1.0 \# artifact tag} \\
\\
\codeIn{python3 -m venv venv} \\
\codeIn{source venv/bin/activate} \\
\codeIn{pip install -U pip} \\
\codeIn{pip install -r requirements.txt}
\end{mdframed}

\MyPara{Set Up the Cluster.}
We now discuss how to setup the cluster with all different roles.
To do this, we use the \codeIn{ec2man} python module.
To start, setup the profile in the following way:
\begin{mdframed}[style=MyFrame,nobreak=true]
\codeIn{\$ vim ec2man/profile} \\
\\
\codeIn{default       \# Profile from \~/.aws/credentials} \\
\codeIn{ubuntu        \# Cluster username} \\
\codeIn{\$\{HOME\}/.ssh/id\_rsa  \# Path to SSH key} \\
\codeIn{us-east-2     \# AWS region}
\end{mdframed}

As mentioned previously, we work with two types of workers which we call
'contexts', specifically graph and weight servers.
To add machines to these two contexts, we use one of the following commands:
\begin{mdframed}[style=MyFrame,nobreak=true]
\codeIn{python -m ec2man allocate --ami [AMI]}

\codeIn{--type [ec2 type] --cnt [\#servers]}

\codeIn{--sg [security group]}

\codeIn{--ctx [weight|graph]}\\
\\
\codeIn{python -m ec2man add [graph|weight]}

\codeIn{[list of ec2 ids]}
\end{mdframed}
Run the first command with an AMI ID that presents a fresh install of Ubuntu,
ideally with about 36 GB of storage.
Alternatively if you have created instances already, say 4 graph servers
you can add them to the module using the \codeIn{add} command with a list of
their IDs.
Finally, run the command \codeIn{python -m ec2man setup} to get the data
about the instances so they can be managed by the module.
To make sure everything is setup correctly, try SSHing into graph server 0
using \codeIn{python -m ec2man graph 0 ssh}.

\MyPara{Building \tool.}
The next step is to make sure all dependencies are installed to build \tool
on the cluster machines.
To do this, run the following commands:
\begin{mdframed}
\codeIn{local\$ ./gnnman/send-source [--force]}

\codeIn{\# '--force' removes existing code} \\
\\
\codeIn{local\$ ./gnnman/install-dep}
\end{mdframed}

This will sync the source code with the nodes on the cluster.
Then, it will install all dependencies required to build \tool.
If this fails for some reason you may need to ssh into each node, move
into the \codeIn{dorylus/gnnman/helpers} directory, and run:
\begin{mdframed}
\codeIn{remote\$ ./graphserver.install} \\
\codeIn{remote\$ ./weightserver.install}
\end{mdframed}

\MyPara{Parameter Files.}
There are a number of parameter files relating to things such as the ports used
during training.
Most of these will be fine as they are and should only be changed if there is a
conflict.

\MyPara{Compiling the Code.}
To build and synchronize the code on all nodes in the cluster run:
\begin{mdframed}
\codeIn{local\$ ./gnnman/setup-cluster} \\
\codeIn{local\$ ./gnnman/build-system}

\codeIn{[graph|weight] [cpu|gpu]}
\end{mdframed}
The first command sets up each node of the cluster to be aware of each other.
This is important as we only build the code on node 0 and distribute it to other
nodes.
The second command runs CMake to build the actual system.
Not specifying a context builds for all contexts.
Not specifying either \codeIn{cpu} or \codeIn{gpu} as the backend builds the
serverless version.

\MyPara{Setting up Lambda Functions.}
To install the Lambda functions, you can SSH into one of the weight or graph
servers.
Once there, run the following commands:
\begin{mdframed}
\# Install the Lambda dependencies \\
\codeIn{remote\$ ./funcs/manage-funcs.install} \\
\\
\# Build and upload the function to the cloud
\codeIn{remote\$ cd src/funcs} \\
\codeIn{remote\$ ./<function-name>/upload-func}
\end{mdframed}

\mysubsubsection{Preparing the Data}
There are 4 main inputs to \tool:
\squishlist
    \item The graph structure
    \item Graph partition info
    \item Input features
    \item Training labels
\squishend

\MyPara{Graph Input.}
To prepare an input graph for \tool, the format should be a binary edge list
with vertices numbered from $0$ to $|V|$ with no breaks using 4 byte values.
The file should be named \codeIn{graph.bsnap}.

\MyPara{Partition Info.}
\tool uses edge-cut partitioning.
While we do limit partitioning to edge-cuts, we allow flexibility in how the edge
cut is implemented by partitioning at runtime.
Provide a text file that lists partition assignments line by line, where each
line number corresponds to the vertex ID and the number is the partition to which
it is assigned.
The file should be called \codeIn{graph.bsnap.parts}.

\MyPara{Input Features.}
The input features take the form of a tensor of size $|V|\times d$ where $d$ is
the number of input features.
The file should be binary and take the format of:
\begin{mdframed}
\codeIn{[numFeats][v0\_feats][v1\_feats][v2\_feats]...}
\end{mdframed}
The file should be called \codeIn{features.bsnap}.

\MyPara{Training Labels.}
The labels file should be binary and take the form:
\begin{mdframed}
\codeIn{[numLabels][label0][label1]...}
\end{mdframed}
This file should be called \codeIn{labels.bsnap}.

\MyPara{Preparing the NFS Server.}
On an NFS server setup the dataset in the following format under a directory
called \codeIn{/mnt/filepool/}.
If the dataset we are preparing is called \codeIn{amazon}, the directory structure
would look like this:
\begin{mdframed}
\codeIn{amazon} \\
\codeIn{|-- features.bsnap} \\
\codeIn{|-- graph.bsnap} \\
\codeIn{|-- labels.bsnap} \\
\codeIn{|-- parts\_<\#partitions>/}

\codeIn{|-- graph.bsnap.edges}

\codeIn{|-- graph.bsnap.parts}
\end{mdframed}
where \codeIn{graph.bsnap.edges} is a symlink to \codeIn{../graph.bsnap}.
Use the \codeIn{add} command from above to add the NFS server to a special
context called \codeIn{nfs} so that \codeIn{ec2man} knows where to look for
it.
Finally, run
\begin{mdframed}
\codeIn{local\$ ./gnnman/mount-nfs-server}
\end{mdframed}

\mysubsubsection{Running Dorylus.}
Once the cluster has been setup, the code compiled, the Lambda functions installed,
and the datasets prepared, we can run Dorylus.
To run it use the following command from the \codeIn{dorylus/} directory on
your local machine:
\begin{mdframed}
\# \codeIn{<dataset>}: the dataset you prepared \\
\# \codeIn{--l}: the \#lambdas/server \\
\# \codeIn{--p}: enable asynchronous pipelining \\
\# \codeIn{--s}: degree of staleness \\
\# \codeIn{[cpu|gpu]}: backend to use (blank means lambda) \\
\\
\codeIn{./run/run-dorylus <dataset>}

\codeIn{[--l=\#lambdas] [--lr=learning\_rate]}

\codeIn{[--p] [--s=staleness] [cpu|gpu]}
\end{mdframed}

You will see the output of the Graph Servers, but can see the output
of both the Graph and Weight Servers in \codeIn{graphserver-out.txt} and
\codeIn{weightserver-out.txt}.
More details of \tool's installation and deployment can be found in
\tool's code repository.
\newpage
\bibliographystyle{abbrv}
\bibliography{paper}

\mysection{Full Proof\label{sec:theorem}}
To simplify the proof without introducing ambiguity, we refer to running asynchronous SGD on GNNs with asynchronous \codeIn{Gather} as \emph{asynchronous GNN training}, while running normal SGD with synchronous \codeIn{Gather} as \emph{synchronous GNN training}.
\mysubsection{Proof of Theorem~\ref{thm:async-gnn}}
Similar to the proof of Theorem 2 in \cite{chen-gcnvr-pmlr18}, we prove Theorem~\ref{thm:async-gnn} in 3 steps:
\begin{enumerate}
    \item Lemma 1: Given a sequence of weights matrices $W^{(1)}, \ldots, W^{(N)}$ which are close to each other, approximate activations in asynchronous GNN training are close to the exact activations.
    \item Lemma 2: Given a sequence of weights matrices $W^{(1)}, \ldots, W^{(N)}$ which are close to each other, approximate gradients in asynchronous GNN training are close to the exact gradients.
    \item Theorem~\ref{thm:async-gnn}: Asynchronous GNN training generates weights that change slow enough for the gradient bias goes to zero, and thus the algorithm converges.
\end{enumerate}

Let $\|A\|_{\infty} = \max_{i, j}\left|A(i, j)\right|$, we first state the following propositions:
\begin{proposition}
\label{prop:a}
\ 
\begin{itemize}
    \item $\|A B\|_{\infty} \leq \operatorname{col}(A)\|A\|_{\infty}\|B\|_{\infty}$, where $col(A)$ is the number of columns of matrix $A$.
    \item $\|A \odot B\|_{\infty} \leq\|A\|_{\infty}\|B\|_{\infty}$, where $\odot$ is element-wise product.
    \item $\|A+B\|_{\infty} \leq\|A\|_{\infty}+\|B\|_{\infty}$.
\end{itemize}
Let 
\begin{align*}
C:= \max& \left\{\operatorname{col}(\hat{A}), \operatorname{col}\left(H^{(0)}\right), \ldots, \operatorname{col}\left(H^{(L)}\right),\right.\\
        & \left.\operatorname{col}\left(W^{(0)}\right), \ldots, \operatorname{col}\left(W^{(L)}\right) \right\}
\end{align*}
be the largest number of columns of matrices we have in the proof, we have
\begin{itemize}
    \item $\|A B\|_{\infty} \leq C\|A\|_{\infty}\|B\|_{\infty}$.
\end{itemize}
\end{proposition}
The proof can be found in Appendix C in \cite{chen-gcnvr-pmlr18}.
\begin{definition}
\label{def:compose}
(Mixing Matrix) We say matrix $\tilde{A}$ is a mixing matrix with $N$ source matrices $\left\{A_1,\ldots, A_N\right\}$ iff. $\forall i, j,\ \exists! k \in [1,N] \textit{ s.t. } \tilde{A}(i,j) = A_k(i, j)$.
\end{definition}
Conceptually, every element in a mixing matrix corresponds to one of the source matrices. And we also have the following proposition:
\begin{proposition}
\label{prop:b}

Suppose that $\tilde{A}$ is a mixing matrices with $A_1,\ldots, A_N$ as its sources, then
$$\left\|\tilde{A}\right\|_{\infty} \leq \max_{k = 1}^{N}\left\{\left\|A_k\right\|_{\infty}\right\}.$$
\end{proposition}
\begin{proof}
By Definition~\ref{def:compose}:
\begin{align*}
\|\tilde{A}\|_{\infty} =& \max_{i, j}\left|\tilde{A}(i, j)\right|\\ 
\leq &  \max_{i, j}\max_{k=1}^{N}\left|A_k(i, j)\right|\\
=&      \max_{k=1}^{N}\max_{i, j}\left|A_k(i, j)\right|\\
=&      \max_{k = 1}^{N}\left\{\left\|A_k\right\|_{\infty}\right\}.
\end{align*}
\end{proof}

\mysubsubsection{A single layer in GCN}
As a base case, we first consider a single layer in a GCN model.
In this case, at epoch $i$, the output activations $H$ depends on weights $W$ and input activations $X$ only.\\
Under synchronous GNN training, we have
\begin{equation*}
Z_{i}=\hat{A} X_{i} W_{i}, \quad H_{i}=\sigma\left(Z_{i}\right).
\end{equation*}
While under asynchronous GNN training with staleness bound $S$, we have
\begin{equation*}
Z_{AS, i}=\hat{A} \tilde{X}_{AS, i} W_{i}, \quad H_{AS, i}=\sigma\left(Z_{AS, i}\right),
\end{equation*}
where $Z_{AS, i}$, $X_{AS, i}$, and $H_{AS, i}$ are corresponding approximate matrices of $Z_i$, $X_i$, and $H_i$. $\tilde{X}_{AS, i}$ is a mixing matrix of $S$ stale activations $X_{AS, i-S+1},\ldots, X_{AS, i}$.

Now we show that the approximate output activations under asynchronous GNN training are close to the exact ones when the weights change slowly during the training.

\begin{proposition}
\label{prop:single-gcn}
Suppose that the activation $\sigma(\cdot)$ is $\rho$-Lipschitz, and for any series of~~$T$ input activations and weights matrices $\left\{X_{i}, W_i\right\}_{i=1}^{T}$, where
\begin{enumerate}
    \item matrices are bounded by some constant $B$: $\left\|\hat{A}\right\|_{\infty} \leq B$,
        $\left\|X_{i}\right\|_{\infty} \leq B$,
        $\left\|X_{AS, i}\right\|_{\infty} \leq B$,
        and $\left\|W_{i}\right\|_{\infty} \leq B$;
    \item differences are bounded by $\epsilon$: $\left\|X_{AS, i}-X_{AS, j}\right\|_{\infty}<\epsilon$,
        $\left\|X_{AS, i}-X_{i}\right\|_{\infty}<\epsilon$,
        and $\left\|W_{i}-W_{j}\right\|_{\infty}<\epsilon$.
\end{enumerate}

Then there exists $K$ that depends on $C$, $B$, and $\rho$, s.t. for all $S < i, j \leq T$, where $S$ is the staleness bound:
\begin{enumerate}
    \item The approximate outputs won't change too fast: $\left\|Z_{AS, i} - Z_{AS, j}\right\|_{\infty} < K\epsilon$, and\\ $\left\|H_{AS, i} - H_{AS, j}\right\|_{\infty} < K\epsilon$,
    \item The approximate outputs are close to the exact outputs: $\left\|Z_{AS, i} - Z_{i}\right\|_{\infty} < K\epsilon$, and\\ $\left\|H_{AS, i} - H_{i}\right\|_{\infty} < K\epsilon$.
\end{enumerate}
\end{proposition}
\begin{proof}
After the $S$ warm up training epochs, we know for all $i > S$, $\tilde{X}_{AS, i}$ consists with either latest neighbor activations (i.e., $X_{AS, i}$) or activations from some previous epochs (i.e., $X_{AS, i-1}, \dots, X_{AS, i-S+1}$), and
\begin{align*}
\left\| \tilde{X}_{AS, i} - X_{AS, i}\right\|_{\infty} 
\leq & \max\limits_{s \in [i-S+1, i)} \left\| X_{AS, i} - X_{AS, s}\right\|_{\infty}\\
\leq & \epsilon.
\end{align*}
By the triangle inequality,
\begin{align*}
\left\|\tilde{X}_{AS, i}-X_{i}\right\|_{\infty} & < 2\epsilon, \forall i \in (S, T]\\
\left\|\tilde{X}_{AS, i}-\tilde{X}_{AS, j}\right\|_{\infty} & < 3\epsilon, \forall i, j \in (S, T]
\end{align*}
By Prop.~\ref{prop:b} and $\left\|X_{AS, i}\right\|_{\infty} \leq B$, we have $\left\|\tilde{X}_{AS, i}\right\|_{\infty} \leq B$.
Thus, $\forall i,j \in (S, T]$, we have
\begin{align*}
& \left\|Z_{AS, i}-Z_{AS, j}\right\|_{\infty}\\
=    & \left\|\hat{A}\tilde{X}_{AS, i}W_i - \hat{A}\tilde{X}_{AS, j}W_j \right\|_{\infty}\\
=    & \left\|\hat{A}\left(\tilde{X}_{AS, i} - \tilde{X}_{AS, j}\right)W_i + 
              \hat{A}\tilde{X}_{AS, j}(W_i - W_j)
    \right\|_{\infty}\\
\leq & C^2\left(
        \left\|
            \hat{A}\right\|_{\infty}\left\|\tilde{X}_{AS, i} - 
            \tilde{X}_{AS, j}\right\|_{\infty}\left\|W_i
        \right\|_{\infty}\right. \\
    &   \left.+ \left\|
            \hat{A}\right\|_{\infty}\left\|\tilde{X}_{AS, j}\right\|_{\infty}\left\|W_i - W_j
        \right\|_{\infty}
    \right)\\
\leq & C^2 (3\epsilon B^2 + \epsilon B^2)\\
=    & 4C^2B^2\epsilon,
\end{align*}
and
\begin{align*}
\left\|Z_{AS, i}-Z_{i}\right\|_{\infty}
= & \left\|\hat{A}\tilde{X}_{AS, i}W_i - \hat{A}X_{i}W_i\right\|_{\infty}\\
= & \left\|\hat{A} \left(\tilde{X}_{AS, i} - X_{i}\right) W_i\right\|_{\infty}\\
\leq & C^2 \left\|\hat{A}\right\|_{\infty}\left\|\tilde{X}_{AS, i} - 
              X_{i}\right\|_{\infty}\left\|W_i\right\|_{\infty}\\
\leq & 2C^2B^2\epsilon.
\end{align*}
By the Lipschitz continuity of $\sigma(\cdot)$,
\begin{align*}
\left\|H_{AS, i} - H_{AS, j}\right\|_{\infty} 
= &     \left\|\sigma(Z_{AS, i}) - \sigma(Z_{AS, j})\right\|_{\infty} \\
\leq &  \rho\left\|Z_{AS, i} - Z_{AS, j}\right\| \\
\leq &  4\rho C^2B^2\epsilon,\\
\left\|H_{AS, i} - H_{i}\right\|_{\infty}
= &     \left\|\sigma(Z_{AS, i}) - \sigma(Z_{i})\right\|_{\infty} \\
\leq &  \rho\left\|Z_{AS, i} - Z_{i}\right\|\\
\leq &  2\rho C^2B^2\epsilon.
\end{align*}
We set $K = \max\left\{4C^2B^2\epsilon, 4\rho C^2B^2\epsilon \right\}$ thus all the differences are bounded by $K\epsilon$.
\end{proof}

\mysubsubsection{Lemma 1: Activations of Multi-layer GCNs}
Now we generalize the conclusion to multi-layer GCNs by applying Prop.~\ref{prop:single-gcn} layer by layer.
The forward computation of a multi-layer GCN can be expressed as follows:\\
Under synchronous GNN training,
\begin{equation}
\label{equ:forward-s}
\begin{aligned}
Z_{i}^{(l+1)}   =& \hat{A} H_{i}^{(l)} W_{i}^{(l)}  & &\quad l=1, \ldots, L-1\\
H_{i}           =& \sigma\left(Z_{i}^{(l+1)}\right) & &\quad l=1, \ldots, L-1
\end{aligned}
\end{equation}
Under asynchronous GNN training,
\begin{equation}
\label{equ:forward-as}
\begin{aligned}
Z_{A S, i}^{(l+1)}=& \hat{A} \tilde{H}_{A S, i}^{(l)} W_{i}^{(l)} & &\quad l=1, \ldots, L-1\\
H_{A S, i}^{(l+1)}=& \sigma\left(Z_{A S, i}^{(l+1)}\right) & &\quad l=1, \ldots, L-1
\end{aligned}
\end{equation}
The following lemma bounds the approximation error of output activations of a multi-layer GCN under asynchronous GNN training.
The Lemma shows that the approximation error is bounded by the change rate of model weights only, regardless of the staleness bound $S$.

\begin{lemma}
\label{lemma:act}
Suppose that all the activations are $\rho$-Lipschitz.
Given any series of T inputs and weights matrices $\left\{H_{i}^{(0)}, W_{i}\right\}_{i=1}^{T}$, where
\begin{enumerate}
    \item $\left\|\hat{A}\right\|_{\infty} \leq B$,
    $\left\|H_{i}^{(0)}\right\|_{\infty} \leq B$, 
    and $\left\|W_{i}\right\|_{\infty} \leq B$,
    \item 
        $\left\|W_{i}-W_{j}\right\|_{\infty}<\epsilon$,
\end{enumerate}

there exists some constant $K$ that depends on $C$, $B$, and $\rho$ s.t. $\forall i > LS$,
\begin{enumerate}
    \item $\left\|H_{i}^{(l)}-H_{AS, i}^{(l)}\right\|_{\infty} < K \epsilon, \quad l=1, \ldots, L$,
    \item $\left\|Z_{i}^{(l)}-Z_{AS, i}^{(l)}\right\|_{\infty} < K \epsilon, \quad l=1, \ldots, L$.
\end{enumerate}
\end{lemma}
\begin{proof}
By Prop.~\ref{prop:single-gcn} and Eq.~(\ref{equ:forward-s},\ref{equ:forward-as}), for all $i > LS$, there exists a constant $K^{(1)}$ s.t. $\left\|H_{i}^{(1)}-H_{AS, i}^{(1)}\right\|_{\infty} < K^{(1)} \epsilon$ and  $\left\|Z_{i}^{(1)}-Z_{AS, i}^{(1)}\right\|_{\infty} < K^{(1)} \epsilon$.

By applying Prop.~\ref{prop:single-gcn} repeatedly for all $L$ layers, we will get $K^{(1)}, K^{(2)} ,\ldots, K^{(L)}$, s.t. 
$$\left\|\!H_{i}^{(2)}\!-\!H_{AS, i}^{(2)}\!\right\|_{\infty}\! <\! K^{(1)}\!K^{(2)}\! \epsilon, \ldots, \left\|\!H_{i}^{(L)}\!-\!H_{AS, i}^{(L)}\!\right\|_{\infty}\! <\! K\! \epsilon, $$
and 
$$\left\|\!Z_{i}^{(2)}\!-\!Z_{AS, i}^{(2)}\!\right\|_{\infty} \!<\! K^{(1)}\!K^{(2)} \!\epsilon, \ldots, \left\|\!Z_{i}^{(L)}\!-\!Z_{AS, i}^{(L)}\!\right\|_{\infty} \!<\! K \!\epsilon,$$
where $K=\prod_{l=1}^{L} K^{(i)}$, for all $i > LS$.
\end{proof}

\mysubsubsection{Lemma 2: Gradients of Multi-layer GCNs}
Denote $f(y, z)$ as the cost function that takes the model prediction $y$ and the ground truth $z$. The gradients in the backpropagation can be computed as follows:\\
Under synchronous training,
\begin{equation}
\label{equ:grad-s}
\begin{aligned}
\nabla_{H^{(l)}} f &=\hat{A}^{T} \nabla_{Z^{(l+1)}} f W^{(l)^T}             & & l=1, \ldots, L-1 \\
\nabla_{Z^{(l)}} f &=\sigma'\left(Z^{(l)}\right) \odot \nabla_{H^{(l)}} f   & & l=1, \ldots, L \\
\nabla_{W^{(l)}} f &=\left(\hat{A} H^{(l)}\right)^{T} \nabla_{Z^{(l+1)}} f  & & l=0, \ldots, L-1
\end{aligned}
\end{equation}
Under asynchronous training,
\begin{equation}
\label{equ:grad-as}
\begin{aligned}
\nabla_{H_{AS}^{(l)}} f_{AS} &= \hat{A}^{T} \nabla_{Z_{AS}^{(l+1)}} f_{AS} W^{(l)^T}                 & & l=1, \ldots, L-1 \\ 
\nabla_{Z_{AS}^{(l)}} f_{AS} &= \sigma'\left(Z_{AS}^{(l)}\right) \odot \nabla_{H_{AS}^{(l)}} f_{AS}  & & l=1, \ldots, L \\ 
\nabla_{W^{(l)}} f_{AS}      &= \left(\hat{A} H_{AS}^{(l)}\right)^{T} \nabla_{Z_{AS}^{(l+1)}} f_{AS} & & l=0, \ldots, L-1
\end{aligned}
\end{equation}
Denote the final loss function as $\mathcal{L}(W_{i})$. Let $g_{i}(W_{i}) = \nabla \mathcal{L}(W_{i})$, which is the gradient of $\mathcal{L}$ with respect to $W_i$ under synchronous GNN training in epoch $i$. And let $g_{AS, i}\left(W_{i}\right) = \nabla \mathcal{L}_{AS, i}(W_{i})$, which is the corresponding gradients of approximate $\mathcal{L}$ under asynchronous GNN training in epoch $i$.
To simplify the proof, we construct the weight update sequences with layer-wise weight updates instead epoch-wise weight updates. But either of them works for the proof.
The following lemma bounds the difference between gradients under asynchronous GNN training and exact ones.
\begin{lemma}
\label{lemma:grad}
Suppose that $\sigma(\cdot)$ and $\nabla_{z} f(y, z)$ are $\rho$-Lipschitz, and $\left\|\nabla_{z} f(y, z)\right\|_{\infty} \leq B$. For the given inputs matrix $H^{(0)}$ from a fixed dataset and any series of $T$ weights matrices $\left\{W_i\right\}_{i=1}^{T}$, s.t.,
\begin{enumerate}
    \item $\left\|W_{i}\right\|_{\infty} \leq B,\left\|\hat{A}\right\|_{ } \leq B, \text { and }\left\|\sigma^{\prime}\left(Z_{AS, i}\right)\right\|_{\infty} \leq B$,
    \item $\left\|W_{i}-W_{j}\right\|_{\infty}<\epsilon, \forall i, j$,
\end{enumerate}

then there exists $K$ that depends on $C$, $B$, and $\rho$ s.t.
$$\left\|g_{AS, i}\left(W_{i}\right) - g_{i}(W_{i})\right\|_{\infty} \leq K \epsilon, \forall i>L S.$$
\end{lemma}
\begin{proof}
By the Lipschitz continuity of $\nabla f$ and Lemma~\ref{lemma:act}, for the final layer $L$, we have
\begin{equation}
\label{equ:L}
\begin{aligned}
\exists K^{(L)}, \textit{s.t. }
\left\|\nabla_{Z_{AS}^{(L)}} f_{AS}-\nabla_{Z^{(L)}} f\right\|_{\infty} 
\leq &  \rho \left\|Z_{AS}^{(L)}-Z^{(L)}\right\|_{\infty} \\
\leq &  \rho K^{(L)}\epsilon.
\end{aligned}
\end{equation}
Besides, by the Lipschitz continuity of $\sigma(\cdot)$ and Lemma~\ref{lemma:act}, $\exists\dot{K}, \textit{ s.t. }\forall l \in [1, L], \left\|\sigma'\left(Z_{AS}^{(l)}\right)-\sigma'\left(Z^{(l)}\right)\right\|_{\infty} \leq 
\rho \dot{K} \epsilon$.\\
We prove by induction on $l$ to bound the difference of $\nabla_{Z_{AS}}f_{AS}$ and $\nabla_{Z}f$ layer by layer in the back-propagation order, i.e., we will prove
\begin{equation}
\label{equ:lemma2a}
\exists K^{(l)}, \forall l \in [1, L], \left\| \nabla_{Z_{AS}^{(l)}} f_{AS}-\nabla_{Z^{(l)}} f \right\|_{\infty} \leq K^{(l)} \epsilon.
\end{equation}
Base case: by Eq.~\ref{equ:L}, the statement holds for $l = L$, where $K^{(L)} = \rho\dot{K}$.\\
Induction Hypothesis (IH):
$$\forall l' > l, \left\|\nabla_{Z_{AS}^{(l')}} f_{AS}-\nabla_{Z^{(l')}} f \right\|_{\infty} \leq K^{(l')} \epsilon.$$
Then for layer $l$, by Eq.~(\ref{equ:grad-s},\ref{equ:grad-as}) and the induction hypothesis,
\begin{align*}
        & \left\| \nabla_{Z_{AS}^{(l)}}f_{AS} - \nabla_{Z^{(l)}}f\right\|_{\infty}\\
=       & \left\| 
            \sigma'\left(Z_{AS}^{(l)}\right) \odot \hat{A}^{T} \nabla_{Z_{AS}^{(l+1)}} f_{AS} W^{(l)^T}\right.\\
        &   \left.- \sigma'\left(Z^{(l)}\right) \odot \hat{A}^{T} \nabla_{Z^{(l+1)}} f W^{(l)^T}
        \right\|_{\infty}\\
\leq    & C\!\left\{ \left\| \!
                \left[\! \sigma'\!\left(\!Z_{AS}^{(l)}\! \right)\! -\! \sigma'\!\left(\!Z^{(l)}\!\right)\! \right]\! \odot\! \left[\!\hat{A}^{T}\! \nabla_{Z_{AS}^{(l+1)}}\! f_{AS}\!\right]\!
                \right\|_{\infty}\! \left\|\!W^{(l)^T}\!\right\|_{\infty}\right.\\
                &\left. + \left\|\! \sigma'\!\left(\!Z^{(l)}\!\right)\! \odot\! \left[\!\hat{A}^{T}\! \left(\!\nabla_{Z_{AS}^{(l+1)}}\! f_{AS}\! -\! \nabla_{Z^{(l+1)}}\! f\!\right)\!\right]\! \right\|_{\infty}\! \left\|\!W^{(l)^T}\!\right\|_{\infty} \right\}\\
\leq    & C^2\!\left\{\!\left\|\! 
                \sigma'\!\left(\!Z_{AS}^{(l)} \!\right)\! -\! \sigma'\!\left(\!Z^{(l)}\!\right)\! \right\|_{\infty}\! \left\|\!\hat{A}^{T} \!\right\|_{\infty}\! \left\|\!\nabla_{Z_{AS}^{(l+1)}}\! f_{AS}\!
                \right\|_{\infty}\! \left\|\!W^{(l)^T}\!\right\|_{\infty}\right.\\
                &\left. + \left\|\! \sigma'\!\left(\!Z^{(l)}\!\right)\! \right\|_{\infty}\! \left\|\! \hat{A}^{T}\! \right\|_{\infty}\! \left\|\! \nabla_{Z_{AS}^{(l+1)}}\! f_{AS}\! -\! \nabla_{Z^{(l+1)}}\! f\! \right\|_{\infty}\! \left\|\!W^{(l)^T}\!\right\|_{\infty}
            \right\}\\
\leq    & C^2\rho \dot{K}\epsilon B^3 + C^2B^2\rho K^{(l+1)}\epsilon B\\
=       & \rho\left(\dot{K}+ K^{(l+1)}\right)B^3C^2\epsilon.
\end{align*}
Thus we set $K^{(l)} = \rho\left(\dot{K}+ K^{(l+1)}\right)B^3C^2$ and equation~(\ref{equ:lemma2a}) holds.

Similarly, we can also bound the difference of $\nabla_{W}f_{AS}$ and $\nabla_{W}f$ in each layer with $K_{W}$:
\begin{equation}
\label{equ:lemma2b}
\exists K_{W}^{(l)}, \forall l \in [0, L-1], \left\| \nabla_{W^{(l)}}f_{AS} - \nabla_{W^{(l)}}f\right\|_{\infty} \leq K_{W}^{(l)}\epsilon.
\end{equation}
By the inequality~(\ref{equ:lemma2a}, \ref{equ:lemma2b}), we have
\begin{align*}
& \left\|g_{AS, i}\left(W_{i}\right) - g_{i}(W_{i})\right\|_{\infty}\\
\leq & \max\limits_{l \in [0, L-1]}\left\|\nabla_{W^{(l)}} f_{AS}-\nabla_{W^{(l)}} f \right\|_{\infty}\\
\leq & K \epsilon,
\end{align*}
where $K = \max\limits_{l \in [0, L-1]}\left\{K_{W}^{(l)}\right\}$.
\end{proof}
\mysubsubsection{Proof of Theorem~\ref{thm:async-gnn}}
\begin{customthm}{\ref{thm:async-gnn}}
Suppose that (1) the activation $\sigma(\cdot)$ is $\rho$-Lipschitz, 
(2) the gradient of the cost function $\nabla_{z}f(y,z)$ is $\rho$-Lipschitz and bounded, 
(3) gradients for weight updates $\left\| g_{AS}(W) \right\|_\infty$, $\left\| g(W)\right\|_\infty$, and $\left\| \nabla \mathcal{L}(W) \right\|_\infty$ are all bounded by some constant $G > 0$ for all $\hat{A}$, $X$, and $W$, 
(4) the loss $\mathcal{L}(W)$ is $\rho$-smooth.
Then given the local minimizer $W^{*}$, there exists a constant $K > 0$, s.t., $\forall N > L\times S$ where $L$ is the number of layers of the GNN model and $S$ is the staleness bound, if we train GCN with asynchronous \codeIn{Gather} under a bounded staleness for $R \leq N$ iterations where $R$ is chosen uniformly from $[1, N]$, we will have
\begin{equation*}
\mathbb{E}_{R}
\left\|\nabla \mathcal{L}\left(W_{R}\right)\right\|_{F}^{2} 
\leq 2 \frac{\mathcal{L}\left(W_{1}\right)-\mathcal{L}\left(W^{*}\right)+K+\rho K}{\sqrt{N}},
\end{equation*}
for the updates $W_{i+1} = W_i - \gamma g_{AS}(W_i)$ and the step size $\gamma = min\left\{ \frac{1}{\rho}, \frac{1}{\sqrt{N}}\right\}$.
\end{customthm}
\begin{proof}
We assume the asynchronous GNN training has run for $LS$ epochs with the initial weights $W_1$ as a warm up, and Lemma~\ref{lemma:grad} holds.
Denote $\delta_{i} = g_{AS, i}\left(W_{i}\right) - \nabla\mathcal{L}\left(W_{i}\right) = g_{AS, i}\left(W_{i}\right) - g_{i}\left(W_{i}\right)$. 
By the smoothness of $\mathcal{L}$ we have
\begin{equation}
\label{equ:L-bound}
\begin{aligned}
& \mathcal{L}\left(W_{i+1}\right)\\
\leq &  \mathcal{L}\left(W_{i}\right) + 
        \left\langle\nabla \mathcal{L}\left(W_{i}\right), W_{i+1}-W_{i}\right\rangle +
        \frac{\rho}{2} \gamma^{2}\left\|g_{AS, i}\left(W_{i}\right)\right\|_{F}^{2} \\
= &     \mathcal{L}\left(W_{i}\right) - 
        \gamma\left\langle\nabla \mathcal{L}\left(W_{i}\right), g_{AS, i}\left(W_{i}\right)\right\rangle +
        \frac{\rho}{2} \gamma^{2}\left\|g_{AS, i}\left(W_{i}\right)\right\|_{F}^{2} \\
= &     \mathcal{L}\left(W_{i}\right) -
        \gamma\left\langle\nabla \mathcal{L}\left(W_{i}\right), \delta_{i}\right\rangle -
        \gamma\left\|\nabla \mathcal{L}\left(W_{i}\right)\right\|_{F}^{2}\\
        & + \frac{\rho}{2} \gamma^{2}\left(
            \left\|\delta_{i}\right\|_{F}^{2} +
            \left\|\nabla \mathcal{L}\left(W_{i}\right)\right\|_{F}^{2}+2\left\langle\delta_{i}, \nabla \mathcal{L}\left(W_{i}\right)\right\rangle
        \right) \\
= &     \mathcal{L}\left(W_{i}\right) -
        \left(\gamma-\rho \gamma^{2}\right)\left\langle\nabla \mathcal{L}\left(W_{i}\right), \delta_{i}\right\rangle\\
        & - \left(\gamma-\frac{\rho \gamma^{2}}{2}\right)\left\|\nabla \mathcal{L}\left(W_{i}\right)\right\|_{F}^{2} +
        \frac{\rho}{2} \gamma^{2}\left\|\delta_{i}\right\|_{F}^{2}.
\end{aligned}
\end{equation}
We firstly bound the inner product term $\left\langle\nabla \mathcal{L}\left(W_{i}\right), \delta_{i}\right\rangle$. For all $i$, consider the sequence of $LS + 1$ weights: $\left\{W_{i-LS}, \ldots, W_{i}\right\}$:
\begin{align*}
&       \max _{i-LS \leq j, k \leq i}\left\|W_{j}-W_{k}\right\|_{\infty}\\
\leq &  \sum_{j=i-LS}^{i-1}\left\|W_{j}-W_{j+1}\right\|_{\infty}\\
= &     \sum_{j=i-LS}^{i-1} \gamma\left\|g_{AS, i}\left(W_{j}\right)\right\|_{\infty}\\
\leq &  \sum_{j=i-LS}^{i-1} \gamma G=LSG \gamma.
\end{align*}
By Lemma~\ref{lemma:grad}, there exists $\hat{K} > 0$, s.t.
\begin{equation*}
\left\|\delta_{i}\right\|_{\infty} = \left\|g_{AS, i}(W_{i}) - g_{i}(W_{i})\right\|_{\infty}
\leq \hat{K}LSG\gamma, \forall i > 0.
\end{equation*}
Assume that $W$ is $D$-dimensional, we have
\begin{align*}
\left\langle\nabla \mathcal{L}\left(W_{i}\right), \delta_{i}\right\rangle
\leq & D^2\left\|\nabla\mathcal{L}(W_i)\right\|_{\infty}\left\|\delta_i\right\|_{\infty} \\
\leq & \hat{K}LSD^2G^2\gamma\\
\leq & K\gamma,
\end{align*}
and
\begin{align*}
\left\|\delta_{i}\right\|_{F}^{2}
\leq & D^2\left\|g_{AS, i}(W_i) + \nabla\mathcal{L}(W_i)\right\|_{\infty}^{2}\\
\leq & D^2\left\|g_{AS, i}(W_i)\right\|_{\infty}^{2} + D^2\left\|\nabla\mathcal{L}(W_i)\right\|_{\infty}^{2}\\
\leq & 2D^2G^2\\
\leq & K,
\end{align*}
where $K = \max\left\{\hat{K}LSD^2G^2\gamma, 2D^2G^2\right\}$.
Apply these two inequality to equation~\ref{equ:L-bound}, we get
\begin{align*}
\mathcal{L}\left(W_{i+1}\right) 
\leq & \mathcal{L}\left(W_{i}\right) + 
\left(\gamma-\rho \gamma^{2}\right) K \gamma \\
& - \left(\gamma-\frac{\rho \gamma^{2}}{2}\right)\left\|\nabla \mathcal{L}\left(W_{i}\right)\right\|_{F}^{2} + 
\rho K \gamma^{2} / 2
\end{align*}
Summing up the above inequalities for all $i$ and rearranging the the terms, we get
\begin{equation*}
\begin{aligned}
& \left(\gamma-\frac{\rho \gamma^{2}}{2}\right) \sum_{i=1}^{N}\left\|\nabla \mathcal{L}\left(W_{i}\right)\right\|_{F}^{2} \\
& \leq    \mathcal{L}\left(W_{1}\right) -
        \mathcal{L}\left(W^{*}\right) + 
        K N\left(\gamma-\rho \gamma^{2}\right) \gamma +
        \frac{\rho K}{2} N \gamma^{2}.
\end{aligned}
\end{equation*}
Divide both sides of the summed inequality by $N\left(\gamma - \frac{\rho\gamma^2}{2}\right)$, and take $\gamma = \min\left\{\frac{1}{\rho}, \frac{1}{\sqrt{N}}\right\}$,
\begin{equation*}
\begin{aligned} 
& \mathbb{E}_{R \sim P_{R}}\left\|\nabla \mathcal{L}\left(W_{R}\right)\right\|_{F}^{2}
= \frac{1}{N}\sum_{i=1}^{N}\left\|\nabla \mathcal{L}\left(W_{i}\right)\right\|_{F}^{2}\\
\leq &  \frac{\mathcal{L}\left(W_{1}\right)-
            \mathcal{L}\left(W^{*}\right)+
            K N\left(\gamma-\rho \gamma^{2}\right) \gamma+
            \frac{\rho K}{2} N \gamma^{2}
        }{N \gamma(2-\rho \gamma)} \\ 
\leq &  \frac{\mathcal{L}\left(W_{1}\right)-
            \mathcal{L}\left(W^{*}\right)+
            K N\left(\gamma-\rho \gamma^{2}\right) \gamma+
            \frac{\rho K}{2} N \gamma^{2}
        }{N \gamma} \\ 
\leq &  \frac{\mathcal{L}\left(W_{1}\right)-
            \mathcal{L}\left(W^{*}\right)
        }{N \gamma}+
        K \gamma(1-\rho \gamma)+
        \rho K \gamma \\ 
\leq &  \frac{\mathcal{L}\left(W_{1}\right)-
            \mathcal{L}\left(W^{*}\right)
        }{\sqrt{N}}+
        K \gamma+
        \rho K / \sqrt{N} \\ 
\leq &  \frac{\mathcal{L}\left(W_{1}\right)-
            \mathcal{L}\left(W^{*}\right)+
            K+
            \rho K
        }{\sqrt{N}}.
\end{aligned}
\end{equation*}
Particularly, we have $\mathbb{E}_{R \sim P_{R}}\left\|\nabla \mathcal{L}\left(W_{R}\right)\right\|_{F}^{2} \rightarrow 0$ when $N\rightarrow \infty$, which implies that the gradient under asynchronous GNN training is asymptotically unbiased.
For simplicity purposes, we only prove the convergence of asynchronous GNN training for GCN here. However, the proof can be easily generalized to many other GNN models as long as they share similar properties in Lemma~\ref{lemma:act} and Lemma~\ref{lemma:grad}.
\end{proof}
\end{sloppypar}

\end{document}